\documentclass[aps,pra,floatfix,nofootinbib,twocolumn]{revtex4-1}
\pdfoutput=1

\usepackage{color}

\usepackage{graphicx}  
\usepackage{dcolumn}   
\usepackage{bm}        
\usepackage{amsmath,amsthm,verbatim,amssymb,amsfonts,amscd}
\usepackage{nicefrac}
\usepackage{braket}
\usepackage{gensymb}
\usepackage{bbm}
\usepackage{dsfont}
\usepackage{verbatim}
\usepackage{rotating}
\usepackage[colorlinks,linkcolor=blue]{hyperref}
\theoremstyle{plain}

\newtheorem{corollary}{Corollary}

\theoremstyle{definition}
\newtheorem*{definition}{Definition}
\newtheorem*{claim}{Claim}

\newcommand\floor[1]{\lfloor#1\rfloor}
\newcommand\ceil[1]{\lceil#1\rceil}

\usepackage[svgnames]{xcolor}
\usepackage{tikz}

\usetikzlibrary{decorations.markings}
\usetikzlibrary{fpu}
\usetikzlibrary{shapes.geometric}
\usetikzlibrary{positioning}
\usetikzlibrary{decorations.pathreplacing,decorations.markings}
\usetikzlibrary{fit}
\pgfdeclarelayer{edgelayer}
\pgfdeclarelayer{nodelayer}

\pgfsetlayers{edgelayer,nodelayer,main}

\tikzstyle{none}=[inner sep=0pt]
\definecolor{hexcolor0x0433ff}{rgb}{0.016,0.200,1.000}

\tikzstyle{physical}=[circle,fill=Black,draw=Black]
\tikzstyle{physical}=[circle,fill=hexcolor0x0433ff,draw=hexcolor0x0433ff]
\tikzstyle{(null)}=[circle,fill=Black,draw=hexcolor0xc4dbff]
\tikzstyle{new}=[circle,fill=White,draw=White]
\tikzstyle{bond}=[thick]
\tikzstyle{bond2}=[thick, opacity=1]
\tikzstyle{smallTensor}=[rectangle,fill=OrangeRed,draw=OrangeRed]
\tikzstyle{redBond}=[draw=OrangeRed]
\tikzstyle{tileBond}=[draw opacity=0.5]

\tikzstyle{brownOpaque}=[color=brown,opacity=0.4, line width=1pt]
\tikzstyle{brown}=[color=brown, line width=1pt]

\definecolor{hexcolor0x0433ff}{rgb}{0.016,0.200,1.000}

\tikzstyle{physical}=[circle,fill=Black,draw=Black]
\tikzstyle{physical}=[circle,fill=hexcolor0x0433ff,draw=hexcolor0x0433ff]
\tikzstyle{tensor}=[rectangle,fill=White,draw=Black, scale=1.5]
\tikzstyle{new}=[circle,fill=White,draw=White]
\tikzstyle{circle2}=[circle,fill=White,draw=Black, scale=0.85]
\tikzstyle{point}=[circle,fill=White,draw=Black, opacity=0, scale=1]

\tikzstyle{smallTensor}=[rectangle,fill=OrangeRed,draw=OrangeRed]
\tikzstyle{symmetry}=[circle,fill=Blue,draw=Blue]
\tikzstyle{redPoint}=[circle,fill=OrangeRed,draw=OrangeRed, scale=0.3]
\tikzstyle{bluePoint}=[circle,fill=hexcolor0x0433ff,draw=hexcolor0x0433ff, scale=0.3]
\tikzstyle{pottsConnect}=[draw=Green, dashed]
\tikzstyle{pottsConnect2}=[thick, draw=green]
\tikzstyle{blackCircle}=[circle,fill=Black,draw=Black, scale=0.1]
\tikzstyle{transparentTile}=[rectangle, draw=Black, fill=white, fill opacity=0, draw opacity=0.5, scale=2.5]
\tikzstyle{gn}=[circle, fill=green, scale=0.3]
\tikzstyle{transparentTile2}=[rectangle, draw=Black, fill=white, fill opacity=0, scale=1.5, line width=0.01mm, draw opacity=0.1]
\tikzstyle{transparentTileBlue}=[rectangle, draw=Black, fill=blue, fill opacity=0.2, scale=2.5, line width=0.01mm]

\tikzstyle{bondBendRight} = [thick, bend right=45]
\tikzstyle{bondBendLeft} = [thick, bend left=45]
\tikzstyle{smallCircle} = [circle, fill=black, scale=0.1]
\tikzstyle{smallTensor2} = [fill=black, scale=0.4]
\tikzstyle{blueBond} = [very thick, draw=blue]

\tikzstyle{arrowUp} = [fill=black, regular polygon, regular polygon sides=3, scale=0.5]
\tikzstyle{arrowDown} = [fill=black, regular polygon, regular polygon sides=3, rotate=180, scale=0.5]
\tikzstyle{arrowLeft} = [fill=black, regular polygon, regular polygon sides=3, rotate=90, scale=0.5]
\tikzstyle{arrowRight} = [fill=black, regular polygon, regular polygon sides=3, rotate=270, scale=0.5]

\tikzset{
  on each segment/.style={
    decorate,
    decoration={
      show path construction,
      moveto code={},
      lineto code={
        \path [#1]
        (\tikzinputsegmentfirst) -- (\tikzinputsegmentlast);
      },
      curveto code={
        \path [#1] (\tikzinputsegmentfirst)
        .. controls
        (\tikzinputsegmentsupporta) and (\tikzinputsegmentsupportb)
        ..
        (\tikzinputsegmentlast);
      },
      closepath code={
        \path [#1]
        (\tikzinputsegmentfirst) -- (\tikzinputsegmentlast);
      },
    },
  },
  mid arrow/.style={postaction={decorate,decoration={
        markings,
        mark=at position .75 with {\arrow[#1]{stealth}}
      }}},
}

\makeatletter
\def\squarecorner#1{
	\pgf@x=\the\wd\pgfnodeparttextbox%
	\pgfmathsetlength\pgf@xc{\pgfkeysvalueof{/pgf/inner xsep}}%
	\advance\pgf@x by 2\pgf@xc%
	\pgfmathsetlength\pgf@xb{\pgfkeysvalueof{/pgf/minimum width}}%
	\ifdim\pgf@x<\pgf@xb%
	\pgf@x=\pgf@xb%
	\fi%
	\pgf@y=\ht\pgfnodeparttextbox%
	\advance\pgf@y by\dp\pgfnodeparttextbox%
	\pgfmathsetlength\pgf@yc{\pgfkeysvalueof{/pgf/inner ysep}}%
	\advance\pgf@y by 2\pgf@yc%
	\pgfmathsetlength\pgf@yb{\pgfkeysvalueof{/pgf/minimum height}}%
	\ifdim\pgf@y<\pgf@yb%
	\pgf@y=\pgf@yb%
	\fi%
	\ifdim\pgf@x<\pgf@y%
	\pgf@x=\pgf@y%
	\else
	\pgf@y=\pgf@x%
	\fi
	\pgf@x=#1.5\pgf@x%
	\advance\pgf@x by.5\wd\pgfnodeparttextbox%
	\pgfmathsetlength\pgf@xa{\pgfkeysvalueof{/pgf/outer xsep}}%
	\advance\pgf@x by#1\pgf@xa%
	\pgf@y=#1.5\pgf@y%
	\advance\pgf@y by-.5\dp\pgfnodeparttextbox%
	\advance\pgf@y by.5\ht\pgfnodeparttextbox%
	\pgfmathsetlength\pgf@ya{\pgfkeysvalueof{/pgf/outer ysep}}%
	\advance\pgf@y by#1\pgf@ya%
}
\makeatother

\pgfdeclareshape{square}{
	\savedanchor\northeast{\squarecorner{}}
	\savedanchor\southwest{\squarecorner{-}}
	
	\foreach \x in {east,west} \foreach \y in {north,mid,base,south} {
		\inheritanchor[from=rectangle]{\y\space\x}
	}
	\foreach \x in {east,west,north,mid,base,south,center,text} {
		\inheritanchor[from=rectangle]{\x}
	}
	\inheritanchorborder[from=rectangle]
	\inheritbackgroundpath[from=rectangle]
}

\begin{document}

\title{Projected Entangled Pair States with continuous virtual symmetries}

\author{Henrik~Dreyer} 
\affiliation{Max-Planck Institute of Quantum Optics, Hans-Kopfermann-Str.~1, 85748 Garching, Germany}
\author{J.~Ignacio~Cirac} 
\affiliation{Max-Planck Institute of Quantum Optics, Hans-Kopfermann-Str.~1, 85748 Garching, Germany}
\author{Norbert~Schuch} 
\affiliation{Max-Planck Institute of Quantum Optics, Hans-Kopfermann-Str.~1, 85748 Garching, Germany}

\begin{abstract}
We study Projected Entangled Pair States (PEPS) with continuous virtual
symmetries, i.e., symmetries in the virtual degrees of freedom,
through an elementary class of models with $\mathrm{SU}(2)$ symmetry.  Discrete
symmetries of that kind have previously allowed for a comprehensive
explanation of topological order in the PEPS formalism. We construct local
parent Hamiltonians whose ground space with open boundaries is exactly
parametrized by the PEPS wavefunction, and show how the ground state can
be made unique by a suitable choice of boundary conditions. 
We also find that these models exhibit a logarithmic correction to the entanglement
entropy and an extensive ground space degeneracy on systems with
periodic boundaries, which suggests that they do not describe
conventional gapped topological phases, but either critical models or some
other exotic phase.
\end{abstract}

\maketitle
\section{Introduction}
Tensor Network States provide an entanglement-based description of
wavefunctions of strongly correlated quantum systems.  Besides being a
powerful numerical tool, they also offer a systematic way of producing
exact representations of quantum wavefunctions with a rich variety of
behaviour \cite{exactRepresentation}. In one dimension they have been used
to completely classify phases resulting from both symmetry breaking and
symmetry-protected topological order \cite{SPT1, SPT2, nachtergaele}. The full
classification of phases in dimensions two and greater is currently still
open, due in some part to the existence of phases with intrinsic
topological order. Simple Projected Entangled Pair State (PEPS)
representations have been
found, for example, for the Toric Code (and all other quantum double
models) \cite{QDoubles}, String-Net models \cite{string-Nets1,
string-Nets2} and fermionic states with chiral topological order
\cite{chiral1, chiral2, chiral3}. The central role in all of these constructions is
played by a \textit{virtual symmetry} of the
tensor, this is, a symmetry in the entanglement degree of freedom.
These symmetries, as well as symmetry twists, are locally undetectable yet
show up in the global topological properties of the system: They allow to
parametrize the ground space manifold, to study anyonic excitations and
their statistics, and to determine the entanglement properties of the
system.

So far, studies of PEPS with virtual symmetries have been restricted
to models with discrete symmetry groups.
Conversely, continuous symmetries are expected to give rise to
qualitatively different behaviour, as is apparent from other areas of
many-body physics. Discrete symmetries can, for example, be spontaneously
broken in less than three spatial dimension at finite temperature while
continuous symmetries cannot \cite{merminWagner}. In the context of
virtual symmetries and topological order, this is particularly appealing:
While the discrete symmetries hitherto studied have given a new
perspective on known models in a tensor network framework, a new type of
symmetry may correspond to unconventional phases beyond current knowledge.

In this paper, we initiate the study of PEPS with a continuous virtual
symmetry.  We focus on the symmetry group $\mathrm{SU}(2)$ and its
fundamental representation and study the most general PEPS with the simplest non-trivial virtual degrees of freedom,
in particular its entanglement properties and the
way in which it appears as a ground state of a local Hamiltonian.
Specifically, we show how we can naturally define a parent Hamiltonian from
the PEPS tensor on a $2\times 2$ patch, and that the PEPS exactly
parametrizes its ground space manifold on any region with open boundary
conditions, an important property known as the \emph{intersection
property} in the PEPS literature \cite{FannesNachtergaeleWerner}.  Subsequently, we show that by a
suitable choice of boundary terms, the parent Hamiltonian can be modified
such as to exhibit a unique ground state in any finite volume.  While this
behavior is closely resemblant to that of PEPS with finite virtual
symmetry group, we find that with periodic boundaries, the ground states
cannot be parametrized purely in terms of symmetry twists, and the system
keeps a ground space degeneracy which is \emph{exponential} in the size of
the boundary. A closer analysis reveals that there are at least two types
of ground states: Those which can be parameterized through symmetry twists and
which span a space of linear dimension in the system size, and a distinct
class of ground states which correspond to extremal ``frozen''
spin configurations which are not coupled to other configurations by the
Hamiltonian, and which contribute an exponential number of states.
Finally, we also study the entanglement properties of PEPS wavefunctions
with virtual $\mathrm{SU}(2)$ symmetry and find that the system exhibits a
logarithmic correction to the area law, as opposed to the constant
correction for known gapped topological phases. Overall, this indicates a
behavior which is clearly distinct of those of gapped topological phases with a
finite number of anyons which exhibit a finite ground space degeneracy and
a constant correction to the entanglement entropy, and indicates either
critical behavior of the system or some unconventional kind of order.  We
discuss possible interpretations of our findings in the conclusions.

\section{The wavefunction}
\label{sec:constuctionOfState}

In this section, we define PEPS with $\mathrm{SU}(2)$ symmetries,
introduce the formalism used for their analysis, and analyse their
entanglement properties.

\subsection{Projected Entangled Pair States}

Projected Entangled Pair States (PEPS) are created from an elementary
5-index tensor $A^i_{uldr}$, with the \emph{physical index} $i = 1, \dots,
d$ (with $d$ the \textit{physical dimension}) and the \emph{virtual
indices} $u,l,d,r = 1, \dots, D$ (with $D$ the \textit{bond dimension}).
The tensor defines a \textit{fiducial state}
\begin{equation}
\label{eq:fiducialstate}
\ket{\psi_{1\times 1}(A)}= \input{tikzpictures/Atensor} = \sum_{iuldr}
A^i_{uldr} \ket{i} \ket{uldr}\ ,
\end{equation}
where the box denotes both the tensor (with the legs the virtual indices)
and the fiducial state.  From the fiducial state, a family of states on the square lattice is generated by means of contraction, 
\begin{equation}
\begin{aligned}
\label{eq:fiducialstate2by1}
\ket{\psi_{2\times 1}(A)} &= \input{tikzpictures/lineWithoutA} \\
&= \bra{\phi^+}_{r_1, l_2} \ket{\psi_{1 \times 1}(A)} \otimes \ket{\psi_{1 \times 1}(A)}
\end{aligned}
\end{equation}
where $\ket{\phi^+} = \sum_{i=1}^D \ket{ii}$, and so further for larger
blocks.

Physical states are obtained by imposing boundary conditions $\ket{X} \in
\left(\mathbb{C}^D \right)^{\otimes (2 N_h + 2 N_v)}$, 
\begin{equation}
\begin{aligned}
\label{eq:imposingBoundary}
\ket{\psi_{2\times 1}(A, X)} &= \input{tikzpictures/imposingBoundary} \\
&= \bra{\phi^+}_{\partial u_1, u_1} \bra{\phi^+}_{\partial u_2, u_2} \dots
\ket{\psi_{2 \times 1}(A)} \ket{X}\ ,
\end{aligned}
\end{equation}
where $\partial u_1, \partial u_2, \dots $ are the indices of $\ket{X}$,
$u_1, u_2, \dots $ are the indices of $\ket{\psi_{2 \times 1}(A)}$ and the
$\bra{\phi^+}$ contract them.  A particular choice are periodic boundary
conditions
\begin{equation}
\begin{aligned}
\label{eq:XperiodicBC}
\ket{X=PBC} &= \ket{\phi^+}_{\partial u_1, \partial d_1} \ket{\phi^+}_{\partial u_2, \partial d_2} \dots
\end{aligned}
\end{equation}

\subsection{A class of SU(2)-invariant PEPS}
\label{sec:constructionSU(2)}

Given a unitary representation $U_g$ of a group $G$, we say that a tensor
is \textit{$G$-invariant} if 
\begin{equation}
\begin{aligned}
\label{eq:gInvariance}
\input{tikzpictures/applyU} &= U_g \otimes U_g \otimes \overline{U_g} \otimes \overline{U_g} \ket{\psi_{1\times 1}(A)} \\
&= \ket{\psi_{1\times 1}(A)} \quad \forall g \in G
\end{aligned}
\end{equation}
Here, the arrows denote the direction in which the $U_g$ and $U_g^\dagger$
act on the tensor $A$.  Previous studies of $G$-invariant PEPS have
focused on discrete groups $G$~\cite{GInjectivity}.  In the following, we
generalize $G$-invariance to the symmetry group $G=\mathrm{SU}(2)$ and
introduce a class of $\mathrm{SU}(2)$-invariant PEPS which we subsequently
study in detail.

We will focus on the case where $D=2$, where $U_g\equiv g$ is the
fundamental representation of $g\in\mathrm{SU}(2)$.
A basis for the two-dimensional
subspace of $(\mathbb{C}^2)^{\otimes 4}$ that is invariant under $U_g
\otimes U_g \otimes \overline{U_g} \otimes \overline{U_g}$ is given by $\{
\ket{w}_{ul} \ket{w}_{dr}, \ket{\phi^+}_{ur} \ket{\phi^+}_{dl} \}$, where
$\ket{w}=\ket{01} - \ket{10}$. Therefore, up to a constant factor the most
general fiducial state $\tilde{A}$ is of the form
\begin{equation}
\label{eq:Atilde}
\begin{aligned}
\input{tikzpictures/AtensorTilde} &= \lambda \ket{0}_p \ket{w}_{ul} \ket{w}_{dr} +  \ket{1}_p \ket{\phi^+}_{ur} \ket{\phi^+}_{dl}
\end{aligned}
\end{equation}
where $\lambda \in \mathbb{C}$, and $\ket{0}$ and $\ket{1}$ are normalized
and linearly independent, but not necessarily orthogonal.

A complication of the tensor $\tilde A$ is that it involves different
entangled states and is not rotationally invariant.  We will now introduce
another tensor $A$ which only requires one kind of entangled state, is
rotationally invariant for $\lambda=1$, yet generates the same
family of states. Specifically, we will show that for any region, there
exists an invertible operator $B$ acting on the virtual indices at the
boundary such that
\begin{equation}
\label{eq:equivalentToInvariantTensor}
\begin{aligned}
\ket{\psi_{N_h \times N_v}(A, X)} = \ket{\psi_{N_h \times N_v}(\tilde{A},
BX)}\ .
\end{aligned}
\end{equation}
In the special case of even $N_h$ and $N_v$ and periodic boundary
conditions (PBC), $B\ket{X=\mathrm{PBC}}=\ket{X=\mathrm{PBC}}$. The price we pay
is that $A$ will no longer be explicitly $\mathrm{SU}(2)$-invariant.
Yet, we will see that using $A$ instead of $\tilde A$ simplifies the
majority of the derivations in this paper.

Specifically, define
\begin{equation}
\label{eq:definition}
\begin{aligned}
\input{tikzpictures/Atensor2} &= \lambda \ket{0}_p \ket{\phi^+}_{ul} \ket{\phi^+}_{dr} +  \ket{1}_p \ket{\phi^+}_{ur} \ket{\phi^+}_{dl} \\
&= \lambda \left | \input{tikzpictures/tile0v8} \right>_p \left | \input{tikzpictures/virtual0v8} \right>_v + \left | \input{tikzpictures/tile1v8} \right>_p \left | \input{tikzpictures/virtual1v8} \right>_v
\end{aligned}
\end{equation}
In the second line we have introduced a notation that we will use
throughout the paper.  The physical states $\ket{0}$ and $\ket{1}$
are depicted by curved lines, whereas straight lines are used to indicate
which qubits are maximally entangled in the virtual space.  $A$ is clearly
rotationally invariant for $\lambda = 1$.

To show Eq.~(\ref{eq:equivalentToInvariantTensor}), note that
\begin{equation}
\label{eq:Atilde2}
\begin{aligned}
\input{tikzpictures/Atensor} &= \input{tikzpictures/AtildeTopRight} = \input{tikzpictures/AtildeDownLeft}
\end{aligned}
\end{equation}
where $Y=\left(\begin{smallmatrix} 0 & 1 \\ -1 & 0 \\
\end{smallmatrix}\right)$ and all
matrices act from left to right and from top to bottom.  Inserting now the
middle form of (\ref{eq:Atilde2}) into the even and the right-hand side
into the odd sublattice of the square lattice, we obtain
\begin{equation}
\label{eq:2by2Equivalence}
\begin{aligned}
\input{tikzpictures/2by2As} &= \input{tikzpictures/2by2AsWithYsWithin} \\
&= \input{tikzpictures/2by2AsNoYsWithin}
\end{aligned}
\end{equation}
which proves (\ref{eq:equivalentToInvariantTensor}) with 
\begin{equation}
\label{eq:BisY}
\begin{aligned}
B=Y \otimes \mathds{1} \otimes Y \otimes \mathds{1} \cdots \otimes
Y^T\otimes\mathds{1}\otimes\cdots
\end{aligned}
\end{equation}
The states generated by $A$ are therefore equivalent to those generated by
the $\mathrm{SU}(2)$-invariant tensor $\tilde A$ up
to invertible boundary terms which can be absorbed into the boundary
conditions $X$. In the following, we will therefore work with the tensor
$A$, and thus restrict to $N_h$, $N_v$ even on PBC.

\subsection{The Loop Picture}
\label{sec:loopPicture}

\begin{figure}[t]
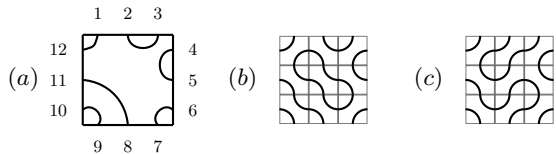

\begin{align*}
(a) \, \input{tikzpictures/3x3ConnectivityPattern} & \quad (b) \input{tikzpictures/loopPattern01} \quad (c) \input{tikzpictures/loopPattern02}
\end{align*}
\caption{
Connectivity patterns and classes.  (a) shows a connectivity
pattern, here 
$\{(1,12),\allowbreak (2,3),\allowbreak (4,5),\allowbreak
(6,7),\allowbreak  (8,11),\allowbreak  (9,10) \}$. (b)~and (c) show two
loop patterns which are compatible with the connectivity
pattern in (a).  All compatible patters form the \textit{connectivity
class} corresponding to (a). For the loop pattern shown in (b), $n_L = 1$
and $b_L = 3$.}
\label{fig:terminology}
\end{figure}

The notation introduced in Eq.~(\ref{eq:definition}) provides a convenient
graphical way of expressing configurations of the PEPS through loop
patterns, by replacing $\ket{0}$ with $\input{tikzpictures/tile0v8Small}$
and $\ket{1}$ with $\input{tikzpictures/tile1v8Small}$, e.g.
\begin{equation}
\label{eq:exampleLoop1}
\begin{aligned}
\begin{small} \left| \begin{matrix} 1 & 0 & 1 & 0 \\ 0 & 1 & 1 & 1 \\ 1 &
0 & 0 & 0 \\ 1 & 1 & 0 & 1 \end{matrix} \right. \Bigg>\end{small} &\rightarrow
\input{tikzpictures/exampleLoopPicture} \end{aligned}
\end{equation}
To each such physical configuration corresponds a configuration of virtual
states, obtained by contracting the tensors with the corresponding
physical states on that patch with open boundaries (note that for
$\langle 0\ket{1}\ne0$, this requires projecting the physical state onto the
\emph{dual} basis vector). Since the virtual $\ket{\phi^+}$ form the same
pattern as the $\input{tikzpictures/tile0v8Small}$ and
$\input{tikzpictures/tile1v8Small}$, and are connected by projecting onto
$\bra{\phi^+}$, which yet again yields $\ket{\phi^+}$, each open loop
corresponds to a virtual state $\ket{\phi^+}$ at the corresponding virtual
indices at the boundary, while each closed loop contributes a factor of
$2$ (due to our choice of normalization).

Let us now rigorously establish such a framework. In the following, we
always consider an $N_h \times N_v$ patch. The degrees of freedom at the
boundary are numbered from $1,\dots,2N$, $N=N_h+N_v$, as shown in
Fig.~\ref{fig:terminology}a.

A \textit{connectivity pattern} $p$ on the boundary of the patch is a
pairing of the numbers $1, \dots, 2N$, $N=N_h+N_v$, into non-crossing
tuples $\{(a_1, b_1), \dots, (a_{N}, b_{N}) \}$, see
Fig.~\ref{fig:terminology}a. 

A \textit{loop pattern} $L$ is a tiling of the patch with tiles
$\input{tikzpictures/tile0v8Small}$ and
$\input{tikzpictures/tile1v8Small}$, such as in
Fig.~\ref{fig:terminology}b,c.  To each loop pattern $L$, there is a
corresponding \textit{loop state} $\ket{L}$ of the physical system, namely
the product state that is obtained by replacing
$\input{tikzpictures/tile0v8Small}$ with $\ket{0}$ and
$\input{tikzpictures/tile1v8Small}$ with $\ket{1}$.  For a loop pattern
$L$, we denote by $n_L$ the number of closed loops in $L$ and $b_L$ is the
number of $\input{tikzpictures/tile0v8Small}$-tiles in $L$.  Each loop
pattern $L$ is \textit{compatible} with a single connectivity pattern,
$p(L)$, namely the one which is obtained by reading off the boundary pairs
which are connected by $L$.  A \textit{connectivity class} $C_p$ for a
given connectivity pattern is the set of all loop patterns which are
compatible with $p$. We will denote the vector space spanned by all loop
states $\ket{L}$ in the connectivity class $C_p$ by $V(C_p)$.

A \textit{boundary matching} is a state on the virtual degrees of freedom
at the boundary corresponding to a connectivity pattern $p=\{(a_1, b_1), \dots, (a_{N},
b_{N}) \}$, this is,
\begin{equation}
\label{eq:boundaryMatching}
\begin{aligned}
\ket{m(p)} = \ket{\phi^+}_{a_1, b_1} \otimes \dots \otimes
\ket{\phi^+}_{a_{N}, b_{N}}\ .
\end{aligned}
\end{equation}

This terminology permits us to write down the wavefunction of our PEPS in a concise way:
\begin{equation}
\label{eq:quantumLoopState}
\begin{aligned}
\ket{\psi_{N_h \times N_v}(A)} = \sum_{\substack{\text{connectivity} \\ \text{patterns} \\ p}} \ket{m(p)} \otimes \sum_{L \in C_p} \ket{L} 2^{n_L} \lambda^{b_L}
\end{aligned}
\end{equation}
The proof is immediate from the definition of the tensors, and the fact
that each closed loop contributes a factor of $2$, as discussed above.

The set of all boundary matchings $\ket{m(p)}$ is linearly independent and
forms a basis of the space of all staggered spin $0$ states (i.e., spin
$0$ up to a action of $Y$ on every second site).  This follows from the
results of Refs.~\cite{Rumer1, Rumer2, Rumer3}, where it is shown that if
one arranges $2N$ spin $\tfrac12$-particles on a circle and considers
non-crossing singlet matching states, these form a minimal basis for the
spin $0$ subspace of the system, together with the fact that $Y$
transforms between singlet matchings and $\ket{\phi^+}$ matchings (the
endpoints are always on different sublattices). This implies two things:
First, we can restrict any boundary condition $X$ to the staggered spin
$0$ space. Second, there exists there exists a dual basis $\{ \bra{m^*(p)}
\}_p$ of that space such that $\braket{m^*(p) | m(q)} = \delta_{pq}$. 

Using the dual basis, we can construct states which are superpositions of
all loop patterns in the same connectivity class, 
\begin{equation}
\label{eq:quantumLoopStateWithDualVector}
\begin{aligned}
\ket{\psi_{N_h \times N_v}(A,X=\ket{m^*(p)})} = 
\sum_{L \in C_p} \ket{L} 2^{n_L} \lambda^{b_L}
\end{aligned}
\end{equation}
For instance, for 
$p = \input{tikzpictures/3x3ConnectivityPatternSmall}$ and $\lambda=1$,
\begin{equation}
\label{eq:exampleLoop2}
\begin{aligned}
\ket{\psi_{3 \times 3}(A, X = \ket{m^*(p)})} &= 2
\input{tikzpictures/loopPattern01Small} +
\input{tikzpictures/loopPattern02Small} + \dots\ .
\end{aligned}
\end{equation}
Moreover, since $\{|m^*(p)\rangle_p\}$ forms a basis of the space of
staggered singlets (and thus of all relevant boundary conditions), we can
express the PEPS obtained from any boundary condition $X$ as
\begin{equation}
\label{eq:quantumLoopStateWithBoundary}
\begin{aligned}
\ket{\psi_{N_h \times N_v}(A,X)} &= \sum_{p} \braket{X|m(p)} \sum_{L \in C_p} \ket{L} 2^{n_L} \lambda^{b_L}
\\
& = \sum_{p} \braket{X|m(p)} 
\ket{\psi_{N_h \times N_v}(A,|m^*(p)\rangle)}\ .
\end{aligned}
\end{equation}

\subsection{Configuration counting}

In the following, we will determine the dimension of the space
\begin{equation}
\label{eq:SNHNV}
\begin{aligned}
\mathcal{S}_{N_h \times N_v} := \text{span} \{ \ket{\psi_{N_h \times N_v}(A,X)} | X \in \mathbb{C}^{2N_h +2N_v} \}
\end{aligned}
\end{equation}
of all physical configurations accessible with our tensor network.  This
will on the one hand be relevant when computing the entanglement entropy
in Sec.~\ref{sec:entanglementEntropy}, and on the other hand when
determining the ground space degeneracy with open boundaries in
Sec.~\ref{sec:parentHamiltonian}. As we have just seen in
Eq.~(\ref{eq:quantumLoopStateWithBoundary}), $\mathcal S_{N_h\times N_v}$ 
is spanned exactly by the states given in (\ref{eq:quantumLoopStateWithDualVector}).

These states are linearly independent -- unless they are zero -- due to
the linear independence of different $\ket{L}$ which follows from the
linear independence of $\ket{0}$ and $\ket{1}$.  In order for 
Eq.~(\ref{eq:quantumLoopStateWithDualVector}) to be non-zero for a given
connectivity pattern $p$, it must hold that $C_p$ is non-empty. We will call
connectivity patterns $p$ for which $C_p$ is empty 
\textit{forbidden}, otherwise we call $p$ \textit{allowed}. 

An example of a forbidden connectivity pattern on a $2 \times 2$ patch is
\begin{equation}
\label{eq:exampleForbidden1}
\begin{aligned}
 \input{tikzpictures/2x2ConnectivityPatternForbidden} 
\end{aligned}
\end{equation}
The intuitive reason for this connectivity class to be empty is the fact
that it request too large amounts of entanglement between the upper and
lower boundary of the system, more than can be mediated by the bulk: The
connectivity pattern requires four maximally entangled states between top
and bottom half, while the PEPS only has two bonds along that cut. If we
tried to find a loop pattern that matches the connectivity pattern, we
would see that the first two north-south connections fill up all available
space:
\begin{equation}
\label{eq:exampleForbidden2}
\begin{aligned}
 \input{tikzpictures/forbiddenWithoutCut4} 
\end{aligned}
\end{equation}

In order to compute the dimension of $\mathcal{S}_{N_h \times N_v}$, we
therefore need to determine the number of allowed connectivity patterns
for that given system size, which we denote by $\mathcal N(N_h,N_v)$. The
explicit form of this number is constructed in Appendix~\ref{sec:numberN},
and given in Eqs.~(\ref{eq:modifiedCatalan}) and
(\ref{eq:exactForm1}). For the main statements of this paper, the
asymptotic behaviour will be sufficient. If we take $N_h$ and $N_v$ to the
thermodynamic limit in a fixed aspect ratio $N_v/N_h=:\alpha-1$, then (as
proven in Appendix~\ref{sec:numberN},
Eq.~(\ref{eq:asymptoticFinal})) the asymptotic behaviour is essentially
that of the Catalan numbers.  Specifically, denoting the size of the
boundary by $N = N_h  + N_v$, we find that $\mathcal N$ scales
asymptotically as
\begin{equation}
\label{eq:asymptoticsNhNv}
\begin{aligned}
\mathcal{N}(\alpha, N) = \frac{4^{N}}{N^{3/2}} \left[ k(\alpha) + \mathcal{O}\left( \frac{1}{N} \right) \right]
\end{aligned}
\end{equation}
with $k(\alpha)$ a function of the aspect ratio that is independent of $N$.

\subsection{Entanglement Entropy}
\label{sec:entanglementEntropy}

We are now ready to determine the scaling behavior of the entanglement in
our model. To this end, consider a partition of the $N_h \times N_v$-torus
into a (small) rectangle $Q$ of size $L_h \times L_v$, and the
(large) rest $R$.
Our goal is to determine the zero Renyi entropy $S_0(\rho_Q)$ of the
reduced state on $Q$, this is, the logarithm of the Schmidt rank of
$\ket{\psi_{N_h \times N_v}(A,\mathrm{PBC})}$ in said partition. To this
end, note that by construction
\begin{equation}
\label{eq_PsiQR}
\ket{\psi_{N_h \times N_v}(A,PBC)} = 
(\Psi_Q\otimes\Psi_R)\ket{\phi^+}^{|\partial Q|}\ ,
\end{equation}
where $\Psi_Q$ is the linear map $\ket{\psi_{L_h\times L_v}(A,\bullet)}$
from the boundary to the bulk in $Q$ and correspondingly for $\Psi_R$, and
the $|\phi^+\rangle^{|\partial Q|}$ are placed along the boundary
between $Q$ and $R$ which has length $|\partial Q|=2L_h+2L_v=:L$.

As we have seen in the preceding section, the map $\Psi_Q$ provides a
bijection between the space $\mathcal V_{\mathrm{allowed}}$ spanned by all
$\ket{m^*(p)}$ with $p$ an allowed matching, and its image $\mathcal S_{L_h\times L_v}$.
$\Psi_R$, on the other hand, provides a bijection between the full
staggered spin $0$ space $\mathcal V_0$ and its image in $R$, as long as
$Q$ is sufficiently small (specifically, if $\min\{N_h, N_v \} >
\frac{3}{2} (L_h + L_v)$), as in that case there are no forbidden
matchings. Intuitively, this follows from the fact that forbidden
matchings arise due to space constraints at the corners, and the region
$Q$ is concave; we provide a proof in
Appendix~\ref{sec:Binvertible}.

Thus, up to these bijections, $\ket{\psi_{N_h \times N_v}(A,\mathrm{PBC})}$ equals
$(\Pi_{\mathcal V_{\mathrm{allowed}}}\otimes \Pi_{\mathcal V_0})
|\phi^+\rangle^{|\partial Q|}$, which has Schmidt rank equal to 
$\mathrm{dim}\,\mathcal V_{\mathrm{allowed}}=\mathcal
N(L_h,L_v)$.  Using (\ref{eq:asymptoticsNhNv}), we obtain that for a fixed
aspect ratio of $Q$, $S_0(\rho_Q)$ scales as 
\begin{equation}
\label{eq:entropyResult}
S_0(\rho_Q) = L \log 2 - \frac{3}{2} \log(L/2) + \log k + \mathcal{O} \left(\frac{1}{L} \right)
\end{equation}
with a non-universal constant $\log k$ that depends on the aspect ratio 
of $Q$.

As expected, the first term corresponds to the area law $|\partial Q|\log
D$, $D$ being the bond dimension of the PEPS. Nevertheless, the
subleading term is logarithmic rather than a constant as in topologically
ordered models.  Such corrections have been investigated in \cite{Moore}.
There, entropies for conformal two-dimensional quantum critical points,
like the quantum dimer \cite{18,19} and the quantum eight-vertex model
\cite{17} have been studied. The authors find a universal logarithmic
correction to the area law which depends on the associated conformal field
theory and the geometrical details of the partition. In particular, the
same theory can have a pure area law for region $A$ being a disk, while
for $A$ rectangular, logarithmic corrections appear. Similarly, in our
case, the notion of allowed and forbidden connectivity patterns -- which is
fundamental to our calculation -- depends on the shape of the partition.
Curiously, the models studied in \cite{Moore} were found to lie at the
boundary of topologically ordered phases \cite{20}.

\section{Parent Hamiltonians}
\label{sec:parentHamiltonian}

In the following, we will study how our $SU(2)$-invariant wavefunction can
appear as a ground state.  To this end, we
will construct a local parent Hamiltonian and subsequently characterize
its ground space, both for open boundary conditions (OBC) and on the torus.  In
particular, we will show that the parent Hamiltonian possesses a property
known as \emph{intersection property} \cite{FannesNachtergaeleWerner, standardForm}, and that we can obtain a unique
ground state with OBC by gapping out the boundary. 

\subsection{Construction of the Hamiltonian and intersection property}
\label{sec:intersection-property}

Parent Hamiltonian for PEPS are constructed by taking a plaquette of spins
and finding a positive operator which annihilates any state that 
looks like the PEPS on that patch, while penalizing orthogonal states. 
In our case, the Hamiltonian acts on a $2\times 2$ plaquette. To this end,
define,
\begin{equation}
\mathcal{S}_{ 2 \times 2} := \text{span} \left. \left\{ \, \input{tikzpictures/parentHv8} \quad \right | X \in (\mathbb{C}^2)^{\otimes 8} \right \}
\end{equation}
and set
\begin{equation}
\label{eq:parentH}
h = \mathds{1} - \Pi_{\mathcal{S}_{2 \times 2}}\ ,
\end{equation}
the projector onto the orthogonal complement of $\mathcal S_{2\times 2}$.
For any patch, we can then define
\begin{equation}
\label{eq:parent-ham-def}
H=\sum_{x,y} h_{(x,y)}
\end{equation}
where $(x,y)$ is the top left spin of $h_{(x,y)}$, and the sum runs over
all $x$ and $y$ on the patch, according to the chosen boundary conditions.
By construction, $H\ge0$ and $h_{(x,y)}\ket{\psi_{N_x\times N_y}(A,X)}=0$
for all $X$, and thus, any $\ket{\psi_{N_x\times N_y}(A,X)}$ is a ground
state of $H$.  The remaining question is thus to understand whether these
states fully span the ground space of $H$. For OBC, this is known as the
\emph{intersection property} (this is, the intersection of the ground
spaces of the $h_{(x,y)}$ is given by the PEPS with arbitrary boundary on
the larger patch).

In order to understand the structure of an arbitrary ground state of $H$,
let us consider the action of $h$ in terms of the loop picture.  It is
convenient to introduce the following notation for loop states
on 2$\times$2 plaquettes:
\begin{equation}
\label{eq:shorthandNotation}
\begin{aligned}
\left| \input{tikzpictures/tile0110}  \right> &= \ket{B}  \\ \\
\left| \input{tikzpictures/tile0010}  \right>  &= \ket{E_1}  & \left| \input{tikzpictures/tile0111}  \right>  &= \ket{E_2} \\
\left| \input{tikzpictures/tile0100}  \right>  &= \ket{E_3}  & \left| \input{tikzpictures/tile1110}  \right> &= \ket{E_4}  \\ \\
\left| \input{tikzpictures/tile1001}  \right>  &= \ket{O_1} & \left| \input{tikzpictures/tile1000}  \right>  &= \ket{O_2}  & \left| \input{tikzpictures/tile1011}  \right> &= \ket{O_3}  \\
\left| \input{tikzpictures/tile1101}  \right>  &= \ket{O_4} & \left| \input{tikzpictures/tile0001}  \right>  &= \ket{O_5}  & \left| \input{tikzpictures/tile0101}  \right> &= \ket{O_6}  \\ 
\left| \input{tikzpictures/tile1010}  \right>  &= \ket{O_7} & \left| \input{tikzpictures/tile0011}  \right>  &= \ket{O_8}  & \left| \input{tikzpictures/tile1100}  \right> &= \ket{O_9} \\
\left| \input{tikzpictures/tile0000}  \right>  &= \ket{O_{10}} & \left| \input{tikzpictures/tile1111}  \right>  &= \ket{O_{11}}
\end{aligned}
\end{equation}
We will refer to $\ket{B}$ as \textit{bubbles}, $\ket{E_i}$ as
\textit{tadpoles}, and $\ket{O_i}$ as \textit{bubble-free states}.
Furthermore, define
\begin{equation}
\label{eq:superposition}
\begin{aligned}
\ket{\phi} = \frac{1}{2\sqrt{2}} \left[ 2\ket{B} + \sum_{i=1}^4 \ket{E_i}\right]
\end{aligned}
\end{equation}
Then, each local term has $12$ possible ground states: 
\begin{equation}
\ket{O_i}\,, \ i=1,\dots,11\ \quad \mbox{and}
\ \ket{\phi}.
\end{equation}
Taking a general state $\ket{g} = \sum_i o_i \ket{O_i} + \sum_i e_i \ket{E_i} + b \ket{B}$, a direct calculation
reveals that $h\ket{g} = 0$ if and only if $e_i = e_j \, \forall i,j$ and $e_i = b/2 \, \forall i$.
This is, in order to be a ground state of $h$, the states $\ket{B}$ and
$\ket{E_i}$ must appear with the relative amplitudes $2:1:1:1:1$,
as in $\ket{\phi}$ -- and this is the
\emph{only} condition in order to be a ground state.

We can thus interpret the Hamiltonian as defining a random walk on the
space of loop configurations, 
\begin{equation}
\label{eq:randomWalk}
\begin{aligned}
2 \, \,\input{tikzpictures/tile0110} \rightleftarrows \input{tikzpictures/tile0010} \rightleftarrows \input{tikzpictures/tile0111} \rightleftarrows \input{tikzpictures/tile0100} \rightleftarrows \input{tikzpictures/tile1110}
\end{aligned}
\end{equation}
i.e., any two states coupled by the transition (\ref{eq:randomWalk}) must
appear in any ground state in superposition with the given relative
amplitude.  Differently speaking, for any orbit of the random walk
(\ref{eq:randomWalk}) acting on all sites, there is at most one ground
state per orbit.  In the following, we will call such a move between loop
configurations a \textit{surgery move} and use the notation $L' =
\sigma(L)$ to describe the fact that loop patterns $L'$ and $L$ are
related by such a move. We will denote \textit{sequences} of surgery moves
by capital letters, e.g. $\Sigma = \sigma_1 \dots \sigma_M$.

We will now use this interpretation to prove that for $H$ on an OBC
rectangle, there is exactly one ground state per connectivity pattern, this
is, the ground space is given by
\begin{equation*}
\mathcal{S}_{N_h \times N_v} := \text{span} \{ \ket{\psi_{N_h \times N_v}(A,X)} | X \in \mathbb{C}^{(2N_h +2N_v)} \}
\end{equation*}
-- this is precisely the intersection property.  In particular, it entails that the degeneracy of
the parent Hamiltonian is given by $\mathcal{N}(N_h, N_v)$.

To start with, note that each surgery move leaves the
connectivity pattern invariant, i.e., $\braket{K|h|L} = 0$ if $K \in C_p
\neq C_q \ni L$. The Hamiltonian is therefore block diagonal in the loop
basis 
\begin{equation} \label{eq:blockDiagonal} 
H=\bigoplus_p H_p \ ,
\end{equation}
where the $H_p$ are supported on $V(C_p)$. Now pick the basis
$\ket{\psi_p}:=\{ \ket{ \psi_{N_h \times N_v} (A, m^*(p)) } \}_p$ 
of $\mathcal S_{N_h\times N_v}$, cf.~Eq.~(\ref{eq:quantumLoopStateWithDualVector}). 
Each of these states is by construction a ground state of $H$, and lives
in the corresponding block $V(C_p)$ of the Hamiltonian. It thus remains to
show that the random walk defined by $H$ couples any two configurations
$L,L'\in C_p$: As argued above, this uniquely fixes the ratios of the
coefficients $\sum_{L\in C_p} c_L \ket{L}$  for any given $p$, which thus
must be equal to those of $\ket{\psi_p}$. (Note that the fact that $\ket{\psi_p}$
is a ground state implies that the ratio must be independent of the chosen
path $\Sigma(L)=L'$ of surgery moves.) Indeed, as we show in
Appendices \ref{app:catalanNumber} and \ref{appendix:uniqueness},
for any given connectivity pattern $p$, we can define a canonical pattern
$L_0$ such that any $L\in C_p$ can be connected to $L_0$ through a
sequence $\Sigma_0$ of surgery moves,
$L_0=\Sigma_0(L)$, and thus, any two $L,L'\in C_p$ are connected through a
sequence $\Sigma$ which goes through $L_0$, 
\begin{equation}
\label{eq:sequenceOfSurgeryMoves}
L =  \Sigma_0^{-1}(\Sigma'_0(L'))\ ,
\end{equation}
where $L_0=\Sigma_0'(L')$.
Differently speaking, the random walk (\ref{eq:randomWalk}) is ergodic in
the space of loop states with a fixed connectivity pattern.  

This implies that (up to normalisation), on a OBC patch of size $N_h\times
N_v$,
\begin{equation}
\label{eq:finalFormOfPsip}
\ket{\psi_p} = \sum_{L \in C_p} 2^{n_L} \ket{L}
\end{equation}
is the unique ground state of $H$ in sector $p$, i.e., $H$ has one ground
state per connectivity pattern $C_p$, and the space of all ground states
is given by $\mathcal S_{N_h\times N_v}$.

\subsection{Open boundary conditions and unique ground state}
\label{sec_OBCMain}
We have just seen that the parent Hamiltonian possesses the intersection
property -- the ground space manifold on any rectangular patch is
precisely given by those configurations which can be obtained by choosing
arbitrary boundary conditions. In the following, we will show that, for $N_h, N_v$ even, it is
possible to gap out the boundary, this is, to add boundary terms to the
parent Hamiltonian which yield a unique ground state.

To this end, we target
\begin{equation}
\label{eq:PsiAsAPEPS}
\begin{aligned}
\ket{\psi} = \input{tikzpictures/uniqueness04}
\end{aligned}
\end{equation}
as the unique ground state, and proceed by constructing its parent
Hamiltonian. In the bulk, the parent Hamiltonian will be the same as
before. On the boundary, however, extra terms appear.  Specifically, we
consider a $2\times 1$ tile 
\begin{equation}
\label{eq_theArc}
\begin{aligned}
\input{tikzpictures/theArc}
\end{aligned}
\end{equation}
at either boundary, and define
\begin{equation}
\label{eq:groundSpaceDeltaH}
\begin{aligned}
\mathcal{R}_{(2n-1,1),(2n,1)} := \text{span} \left. \left\{ \, \input{tikzpictures/uniqueness05} \quad \right | X \in (\mathbb{C}^2)^{\otimes 4} \right \}
\end{aligned}
\end{equation}
(and rotated versions thereof) and the corresponding parent Hamiltonian
\begin{equation}
\label{eq:boundaryHamiltonianR}
\begin{aligned}
h'_{(x_1,y_1),(x_2,y_2)} = \mathds{1} - \Pi_{\mathcal{R}_{(x_1,y_1),(x_2,y_2)}}\ .
\end{aligned}
\end{equation}
It is easy to check that $h'$, together with the original parent Hamiltonian on the
corresponding $2\times 2$ patch, has exactly the same ground space
as the ``true'' parent Hamiltonian derived from that patch of 
$\ket\psi$ \emph{including} the boundary condition (and containment, which suffices
for $\ket\psi$ to be a ground state, holds trivially). On the other hand,
the parent Hamiltonians on the shifted patches remain unchanged. Thus,
\begin{align}
H' &:= H + \sum_{n=1}^{N_h/2} [h'_{(2n-1,1),(2n,1)} + h'_{(2n-1,N_v),(2n,N_v)}] 
\nonumber
\\
\label{eq:newHamiltonianR}
&\quad+  \sum_{n=1}^{N_v/2} [h'_{(1,2n-1),(1,2n)} + h'_{(N_h,2n-1),(N_h,2n)}]
\end{align}
is a parent Hamiltonian of $\ket\psi$, and has $\ket\psi$ as a ground
state.

Let us now show that this ground state is unique. To this end, note that
ground states of $h'$ on a $2\times 1$ patch are spanned by the states
\begin{equation}
\label{eq:expansionTopLeft}
\begin{aligned}
\ket{\theta_1} &=  \input{tikzpictures/top00} +
\input{tikzpictures/top11} + 2 \input{tikzpictures/top10} \ , \\
\ket{\theta_2} &= \input{tikzpictures/top01} \ .
\end{aligned}
\end{equation}
Thus, $h'$ imposes the additional constraint that in any ground state the
states in $\ket{\theta_1}$ must appear as superpositions with the
corresponding weights.  Arguing as before, this fixes the relative
amplitudes of any two loop patterns coupled by the additional surgery move
\begin{equation}
\label{eq:randomWalk2}
\begin{aligned}
\input{tikzpictures/top00} \rightleftarrows \input{tikzpictures/top11}
\rightleftarrows 2\input{tikzpictures/top10} 
\end{aligned}
\end{equation}
on the corresponding $2\times 1$ patches, and rotated versions thereof.

As we show in Appendix~\ref{app:obc-unique}, any loop pattern can be
transformed to a loop pattern in the ``minimal'' connectivity class
\begin{equation}
\label{eq_pMinimal}
\begin{aligned}
p_\text{min}=\input{tikzpictures/uniqueness01}
\end{aligned}
\end{equation}
by using the original (bulk) surgery moves, together with the additional
surgery move (\ref{eq:randomWalk2}) obtained from $h'$ (which allows to
change the connectivity class). On the other hand, we have seen in
Sec.~\ref{sec:intersection-property} that any two loop patterns in a given
connectivity class -- specifically, the minimal one above -- are connected
through bulk surgery moves.  Thus, it follows that any two loop patterns
can be connected by combining bulk and boundary surgery moves, and thus,
the relative amplitudes of all loop patterns are fixed and therefore equal
to those found in $\ket\psi$, Eq.~(\ref{eq:PsiAsAPEPS}). We thus infer that
$\ket\psi$ is the unique ground state of $H'$.

\subsection{Periodic Boundary Conditions}

Let us now study the ground space structure of the parent Hamiltonian
(\ref{eq:parent-ham-def}) on a system with periodic boundary conditions
(PBC);  recall from Sec.~\ref{sec:constuctionOfState} that this requires
$N_h$ and $N_v$ to be even.  To this end, we will resort to the description of the PEPS in
terms of the tensor $\tilde{A}$, Eq.~(\ref{eq:Atilde}), rather than $A$
(see Sec.~\ref{sec:constructionSU(2)}).  Note that due to the gauge
relation (\ref{eq:equivalentToInvariantTensor}) between them, both $A$ and
$\tilde A$ have the same parent Hamiltonian as defined
in~(\ref{eq:parentH}).

Let us first consider an approach which allowed to fully
characterize the ground space for $G$-injective PEPS with finite symmetry
group $G$~\cite{GInjectivity}.  (In the following, all arrows point from
left to right and top to bottom by convention.) First, note that the
fundamental symmetry (\ref{eq:gInvariance}) is stable under concatenation,
e.g.
\begin{equation}
\begin{aligned}
\label{eq:concatenation}
\input{tikzpictures/lineAround2}  &=  \input{tikzpictures/lineWithUs2}  \\ &= \input{tikzpictures/lineWithout}
\end{aligned}
\end{equation}
i.e., any closed loop of symmetry operators leaves a simply connected
patch invariant. This is particularly interesting when we consider closed boundary conditions:
\begin{equation}
\begin{aligned}
\label{eq:movingStrings}
\input{tikzpictures/cylinderLeft2} &= \input{tikzpictures/cylinderWithUs2} &= \input{tikzpictures/cylinderRight2}
\end{aligned}
\end{equation}
Virtual string operators of the form $U_g^{\otimes N_v}$ 
which wrap vertically around the torus can therefore be freely moved around the
torus, and correspondingly horizontal loops $V_h^{\otimes N_h}$, i.e., the state
\begin{equation}
\begin{aligned}
\label{eq:stringInsertedPsi}
\ket{\psi_{N_h \times N_v}\{U_g,V_h\}} = \, \input{tikzpictures/string_inserted}
\end{aligned}
\end{equation}
on the torus is independent of the position of the strings, as long as $[U_g,V_h] = 0$ 
(otherwise, the strings might not be movable where they intersect).

It is now clear that any such state $\ket{\psi_{N_h \times
N_v}\{U_g,V_h\}}$ is a ground state of the parent Hamiltonian $H=\sum
h_{(x,y)}$, since for any local term $h_{(x,y)}$, the strings can be moved
such that they are outside the region where $h_{(x,y)}$ acts. In the case
of $G$-injective PEPS with finite symmetry group, it could be shown that
these states precisely parameterize the full ground space of
$H$~\cite{QDoubles}. For
abelian groups, all $(g,h)$ yield linearly independent ground states 
$\ket{\psi_{N_h \times N_v}\{U_g,V_h\}}$, while for non-abelian groups,
linear dependencies arise as certain $(g,h)$ yield identical states.

Let us now consider the case of $G=\mathrm{SU}(2)$. Clearly, 
\begin{equation}
\begin{aligned}
\label{eq:stringInsertedState}
\mathcal S' = \mathrm{span} \left \{ \ket{\psi_{N_h \times N_v}\{U,V\}} \, \middle| \, U,V \in SU(2), [U,V]=0 \right \}
\end{aligned}
\end{equation}
is inside the ground space of $H$. What is the dimension of
$\mathcal S'$? Without loss of generality, we can restrict to $U
=\text{diag}(e^{i \phi}, e^{-i \phi})$ -- otherwise, we conjugate each
$\tilde A$ with the the unitary which diagonalises $U$, leaving the state
invariant. Then (up to basis permutations), 
\begin{equation}
\begin{aligned}
\label{eq:stringInsertedState}
U^{\otimes N_v} &= e^{iN_v \phi} \mathds{1}_{\binom{N_v}{0}} \oplus e^{i (N_v - 2) \phi} \mathds{1}_{\binom{N_v}{1}} \oplus \\ & \dots \oplus e^{-iN_v \phi} \mathds{1}_{\binom{N_v}{N_v}},
\end{aligned}
\end{equation}
for arbitrary values of $\phi$, and thus, the closure $U_g^{\otimes N_v}$
on its own parametrizes a $(N_v+1)$--dimensional subspace (e.g.\ by
choosing Fourier angles $\phi_k=2\pi k/(2N_v+1)$, $k=-N_v/2,\dots,N_v/2$).
In order to satisfy $[U,V]=0$, we must have $V=\mathrm{diag}(e^{i \theta},
e^{-i \theta})$, and thus, $\mathcal S'$ 
is at most $(N_h+1)(N_v+1)$-dimensional.  However, it is easy to see that
there is at least one more redundancy: By conjugating each $\tilde A$ with the Pauli
$X$-operator, we map $\phi\to-\phi$, $\theta\to-\theta$. This reduces the number
of possibilities by a factor of $2$, except at $\phi=\theta=0$. 
As we show in Appendix~\ref{app:strings}, the remaining states are
indeed linearly independent, and thus, 
\[
\mathrm{dim}\,\mathcal S' =  \frac{(N_h+1)(N_v+1)+1}{2}\ .
\]

One might think that this parameterizes the full ground space of $H$, just
as for $G$-injective PEPS with finite $G$.  However, this is not the case.
To see this, consider an arbitrary bit-string $b \in \{0,1\}^{N_h}$. Then,
we define the product state $\ket{v(b)}$ by stacking $N_v$ copies of $b$
on top of each other and then identifying $0 \rightarrow \ket{0}$ and $1
\rightarrow \ket{1}$, for example
\begin{equation}
\begin{aligned}
 	 \ket{v(0101)} = \raisebox{0.8cm} {\begin{turn}{-90}
	 \input{tikzpictures/indistinguishable1copy}
	 \end{turn}}
	\label{fig:vb}
\end{aligned}
\end{equation}
Horizontally stacked states $\ket{h(b)}$ are defined accordingly.
Clearly, there are $2^{N_h} + 2^{N_v}-2$ of these states 
(since only the all-0 and all-1 states are doubly counted).
Finally, all of them are ground states, since, by definition, no plaquette
locally looks like any of $\ket{B} = \tiny \ket{\begin{matrix} 0 \, 1 \\ 1
\, 0 \\ \end{matrix}}$, $\ket{E_1} = \tiny \ket{\begin{matrix} 0 \, 0 \\ 1
\, 0 \\ \end{matrix}}$, $\ket{E_2} = \tiny \ket{\begin{matrix} 0 \, 1 \\ 0
\, 0 \\ \end{matrix}}$, $\ket{E_3} = \tiny \ket{\begin{matrix} 1 \, 1 \\ 1
\, 0 \\ \end{matrix}}$ or $\ket{E_4} = \tiny \ket{\begin{matrix} 0 \, 1 \\
1 \, 1 \\ \end{matrix}}$, even across the boundary. Note that in all of
these configurations, the winding of the loops around the torus is maximal
in at least one direction (horizontally or vertically). We will call these
states \textit{isolated states}, as they are not coupled to any other loop
configuration by the Hamiltonian. 

We therefore find that the ground space degeneracy of $H$ is
at least \emph{exponential} in $N_v$ and $N_h$, and thus cannot be parametrized by
strings of symmetry operations alone.  In fact, e.g. the states
\begin{equation}
\begin{aligned}
\label{eq:indistinguishables}
\ket{v(0101)} = \raisebox{0.8cm} {\begin{turn}{-90} \input{tikzpictures/indistinguishable1} \end{turn}} \qquad \ket{v(1010)} = \raisebox{1.05cm} {\begin{turn}{-90} \reflectbox{\input{tikzpictures/indistinguishable1} }\end{turn}}
\end{aligned}
\end{equation}
are indistinguishable by any such string operation.  It is worth pointing
out, however, that all of these ground states are isolated and in the
sector with maximal winding number, so it might still be possible that in
the remaining sectors, the ground space can be parametrized succinctly in
terms of the symmetry.

\section{Conclusions and outlook}

In this paper, we have studied PEPS with continuous virtual symmetries.
Specifically, we have considered the class of $\mathrm{SU}(2)$--invariant
PEPS with the fundamental represention of $\mathrm{SU}(2)$, and studied
their entanglement properties and their relation to local Hamiltonians.
First, we have introduced the most general form of tensors invariant under the
fundamental representation of $\mathrm{SU}(2)$. From the local tensor, we
have constructed local parent Hamiltonians acting on $2\times 2$ sites,
and characterized their ground space structure.  For open boundaries, we
have found that the ground space on rectangular patches on any size is
always exactly parameterized by the PEPS, a property known as the
intersection property.  We were further able to show that by choosing
appropriate Hamiltonian terms at the boundary, the system acquires a
unique ground state.  On a system with periodic boundary conditions, we
have found a ground space degeneracy which grows with the system size.  We
were able to attribute this to at least two distinct mechanisms: First,
closing the boundaries with symmetry twists of $\mathrm{SU}(2)$, in
analogy to finite symmetry groups, yields a linearly growing number of
ground states; and second, extremal isolated spin configurations yield an
exponentially growing number of states. Regarding the entanglement
properties of the state, we found that the zero Renyi entropy has a
logarithmic correction to the area law scaling.

The observed results are clearly distinct from those found for known
topologically ordered phases, and point towards a critical nature of the
wavefunction (due to the logarithmic correction to the entanglement
entropy and the algebraic ground space degeneracy, if the isolated states
are ignored), or possible some other exotic phase. Interestingly, for
$\lambda=1$ and an orthogonal physical site basis, $\langle 0\ket1_p=0$, the
normalization of the loop model corresponding to the
$\mathrm{SU}(2)$--invariant PEPS can be mapped to the partition function
of the $Q=16$ state Potts model at the phase transition between ordered
and disordered phase (see Appendix~\ref{sec:PottsModel}), which is known
to be a first-order transition with a finite correlation length; however,
the mapping implies exponential decay of diagonal observables only.
It may still be the case that other observables exhibit critical correlations
(as is the case for the plaquette-flip for a related quantum loop model \cite{Qgroup5}).
Thus, further studies might be required to determine the precise nature
of the wavefunctions considered, and the way in which it is affected by the
choice of $\lambda$ and the local basis.

\acknowledgements
We are grateful for illuminating discussions with P.~Fendley,
F.~Pollmann and R.~Verresen. We acknowledge the
hospitality of the Centro de Ciencias de Benasque Pedro
Pascual where part of the work was done. This work was supported by the
European Union through the ERC grants WASCOSYS (No.~636201) and QENOCOBA (No.~742102).

\onecolumngrid

\vspace{.8cm}
\hspace{\fill}\rule{12cm}{.3mm}\hspace{\fill}
\vspace{.8cm}

\twocolumngrid

\appendix
\section{Allowed Connectivity Patterns and the Canonical Loop Pattern}
\label{app:catalanNumber}
Both the degeneracy of the the parent Hamiltonian and the entanglement entropy depends on the number $\mathcal{N}(N_h, N_v)$, i.e. the number of connectivity patterns on an $N_h \times N_v$-patch which have at least one compatible loop pattern. In order to show that this number is equal to $\mathcal{N}(N_h, N_v) \sim \frac{4^N}{N^{3/2}}$ in appendix \ref{sec:numberN}, we need to introduce the following terminology:

\begin{definition}\textbf{(Lattice \& Boundary)} \\
Define
\begin{equation}
\begin{aligned}
\mathcal{X} &:=  \{ \frac{1}{2}, \frac{3}{2}, \dots, \frac{2N_h -1}{2} \} \times \{ 0,1, \dots, N_v \} \\
\mathcal{Y} &:=   \{ 0,1, \dots, N_h \} \times \{ \frac{1}{2}, \frac{3}{2}, \dots, \frac{2N_v -1}{2} \}
\end{aligned}
\end{equation}
The $(N_h, N_v)$\emph{-Lattice} is defined as
\begin{equation}
\begin{aligned}
\mathcal{L}_{N_h, N_v} &= \mathcal{X} \cup \mathcal{Y}
\end{aligned}
\end{equation}
The \emph{boundary} $\mathcal{B}_{N_h, N_v} \subset \mathcal{L}_{N_h, N_v}$ is
\begin{equation}
\begin{aligned}
\mathcal{B}_{N_h, N_v} = \bigg\{ (x,y) \in \mathcal{L}_{N_h, N_v} \, | \, x \in \{ 0, N_h \} \text{ or } y \in \{ 0,N_v \}  \bigg\}
\end{aligned}
\end{equation}
\end{definition}

\begin{definition}\textbf{(Tuple Distance)} \\
Let $a,b \in \mathcal{L}_{N_h, N_v}, a \neq b$ and writing $a=(a_x, a_y), b=(b_x, b_y)$, the \emph{x-distance} (\emph{y-distance}) of the tuple $(a,b)$ is
\begin{equation}
\begin{aligned}
\Delta x(a,b) &= b_x - a_x \\
\Delta y(a,b) &= b_y - a_y
\end{aligned}
\end{equation}
A tuple $(a,b)$ is
\begin{equation}
\begin{aligned}
\textit{horizontal}, \text{if } \,  |\Delta x(a,b)| & > |\Delta y(a,b)|,  \\
\textit{vertical}, \text{if } \,  |\Delta x(a,b)| &< |\Delta y(a,b)| \text{ and} \\
\textit{diagonal}, \text{if } \,  |\Delta x(a,b)| &= |\Delta y(a,b)|. \\
\end{aligned}
\end{equation}
A horizontal tuple $(a,b)$ is \emph{upper} if $a_y + b_y \leq N_v$, otherwise it is \emph{lower}. A vertical tuple $(a,b)$ is \emph{left} if $a_x + b_x \leq N_h$, otherwise it is \emph{right}.
\end{definition}

In the main text, we define allowed and forbidden matchings by the existence of at least one compatible loop pattern. We will now give a more useful definition in terms of \textit{Flow} and then show that the definitions are equivalent, i.e. show that for each allowed connectivity pattern as defined here there exists at least one loop pattern: the canonical loop pattern. The fact that there cannot exist a loop pattern for forbidden matchings is easy to see.

\begin{definition}\textbf{(Flow)} \\
For $i \in \{1, N_h-1\}$ ($i \in \{1, N_v-1\}$) and $a,b \in \mathcal{L}_{N_h, N_v}$, the tuple $(a,b)$ \emph{goes through vertical (horizontal) cut i} if $(a,b)$ is horizontal (vertical) and 
\begin{equation}
\begin{aligned}
a_x < i \quad & \text{and} \quad b_x > i \\
(a_y < i \quad & \text{and} \quad b_y > i)
\end{aligned}
\end{equation}
For $p$ a connectivity pattern, the \emph{flow through vertical (horizontal) cut i}, denoted by $\textit{Flow}(p,i,\textit{vert})$ ($\textit{Flow}(p,i,\textit{hor}$)) is the number of bonds $t \in p$ that go through vertical (horizontal) cut $i$. 
\end{definition}

\begin{definition}\textbf{(Forbidden Matchings)} \\
We call a connectivity pattern $p$ \emph{vertically forbidden} if there exists $i \in \{1, 2, \dots, N_h-1\}$ such that
\begin{equation}
\begin{aligned}
\textit{Flow}(p,i,\textit{vert}) \geq N_v + 1
\end{aligned}
\end{equation}
or $\emph{horizontally forbidden}$ if there exists $i \in \{1, 2, \dots, N_v-1\}$ such that
\begin{equation}
\begin{aligned}
\textit{Flow}(p,i,\textit{hor}) \geq N_h + 1
\end{aligned}
\end{equation}
A connectivity pattern which is not forbidden, is \emph{allowed}.
\end{definition}

\begin{definition}\textbf{(The Canonical Loop Pattern)}\\
Given an allowed connectivity pattern $p=\{(a_1, b_1), \dots, (a_N, b_N)\}$, we construct the loop pattern explicitly:
\begin{enumerate}
\item We start with the empty loop pattern $L=\{ \}$.
\item \emph{(Initial and final pieces)} For each $t_i = (a_i, b_i)$, determine whether it is horizontal, diagonal or vertical. If $t_i$ is horizontal or diagonal and $a_i (b_i) \in \mathcal{X}$ then define $\bar{a_i} = a_i + (\nicefrac{1}{2}, \pm \nicefrac{1}{2})$ and $\bar{b_i}  = b_i + (- \nicefrac{1}{2}, \pm \nicefrac{1}{2})$, depending on whether $a_i (b_i)$ are located on the top or bottom boundary. Similarly, if $t_i$ is vertical and $a_i (b_i) \in \mathcal{Y}$, define $\bar{a_i} = a_i + (\pm \nicefrac{1}{2}, \nicefrac{1}{2})$ and $\bar{b_i} = b_i + (\pm \nicefrac{1}{2}, - \nicefrac{1}{2})$. Else, just set $\bar{a_i} = a_i$ and $\bar{b_i} = b_i$. This causes all horizontal and diagonal bonds effectively go from $\mathcal{Y}$ to $\mathcal{Y}$ and all vertical bonds to go from $\mathcal{X}$ to $\mathcal{X}$.
\item \emph{(Choosing a bond)} Pick a bond $t=(a,b) \in p$, such that all bonds inside $t$ have been picked already. Since two path cannot be mutually inside each other, there always exists such a path, except if all bonds have been chosen. In that case, continue with step 8. Define a new partial path $m = \{(a), \bar{a}\}$ (the brackets indicate to only add $a$ if $a \neq \bar{a}$).
\item Set $j=1$ and $v_1 = \bar{a}$.
\item If $v_j = \bar{b}$, add the completed path $m = \{(a), \bar{a}, v_2, \dots, \bar{b}, (b) \}$ to $L$ and go back to step 3. Otherwise continue with the step 6.
\item \emph{(Diagonal partial paths)} If the pair $(v_j,\bar{b})$ is diagonal, consider without loss of generality the case where $x(\bar{b}) > x(v_j)$ and $y(\bar{b}) > y(v_j)$. Then, set $v_{j+1} = v_j + (\nicefrac{1}{2}, \nicefrac{1}{2})$. In the other cases, extend the path towards $\bar{b}$ analogously.
In principle, $v_{j+1}$ could already be occupied by a path $q$. However, as will become clear in the next step, all paths are constructed \emph{monotonously}, i.e. horizontal paths advance towards the right in each step and vertical paths advance towards the bottom. Therefore one can draw a horizontal (vertical) cone if $q$ is horizontal (vertical) and the endpoints, lets call them $a_q$ and $b_q$ must lie inside the cone as well, one to the right (top) of $v_{j+1}$ and one to the left (bottom). It is easy to see that the bonds $(a_q, b_q)$ and $(a,b)$ are crossing, violating the assumption that $m$ is a valid matching.\\
Add $v_{j+1}$ to $p$, set $j \leftarrow j+1$ and go back to step 5.
\item \emph{(All other types of partial paths)} If $(v_j,\bar{b})$ is not diagonal, consider without loss of generality $(a,b)$ to be a lower horizontal bond (all other cases follow analogously). By construction (see below), at any point $(v_j,\bar{b})$ remains horizontal. Define
\begin{equation}
\begin{aligned}
s^{in}_1 &= v_j + (\nicefrac{1}{2}, \nicefrac{1}{2}) \\ s^{out}_1 &= v_j + (\nicefrac{1}{2}, -\nicefrac{1}{2}) \\
s^{in}_2 &= s^{in}_1 + (\nicefrac{1}{2}, \nicefrac{1}{2}) \\ s^{out}_2 &= s^{out}_1 + (\nicefrac{1}{2}, -\nicefrac{1}{2})
\end{aligned}
\end{equation}
If neither of $s^{in}_1$ and $s^{in}_2$ is occupied or in the boundary, set $v_{j+1} = s^{in}_1$ and $v_{j+2} = s^{in}_2$. Otherwise, set $v_{j+1} = s^{out}_1$ and $v_{j+2} = s^{out}_2$. Add $v_{j+1}$ and $v_{j+2}$ to $p$, set $j \leftarrow j+1$ and go back to step 5. \\
Again, in principle one of $s^{in}_1$ and $s^{in}_2$ \emph{and} one of $s^{out}_1$ and $s^{out}_2$ could be occupied or in the boundary. We are now going to show that in this case $p$ is forbidden, i.e. there is too much flow going through a horizontal/vertical line.\\
First, let us assume that $s^{in}_2$ is occupied. Denote by $q$ the path that contains $s^{in}_2$ and call its endpoints $(a_q, b_q)$. Then $q$ must be horizontal, which can be verified using the fact that $(v_j, b)$ is horizontal. As such, $v_j + (\nicefrac{1}{2},\nicefrac{3}{2}) \in q$, since $v_j$ is still free. Now we have two horizontal paths, $p$ and $q$, both go through $x(v_j)$ and their vertical distance at that point is 2. By construction, the vertical distance must remain even all the way through to the initial and final points of paths $q$ and $m$, which implies that there is an odd number of boundary points between $a_m$ and $a_q$ and between $b_m$ and $b_q$. Hence, there is one horizontal bond $(a_r, b_r)$ that goes through $x(s^{in}_2)$ and $(a_q, b_q)$ lies inside it. \\
Consider now $s^{in '}_1 = s^{in}_1 + (0,1)$ and $s^{in'}_2 = s^{in}_2 + (0,1)$. If either of them are in the boundary, the situation is depicted in figure \ref{fig:touchingboundary}.
\begin{figure}[!ht]
  \caption{$s^{in^{'}}_{1}$ touching the boundary}
  \centering
  \label{fig:touchingboundary}
    \input{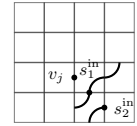}
\end{figure}
Otherwise, consider the progression from $v_j + (\nicefrac{1}{2},\nicefrac{3}{2})$ to $s^{in}_2$: It is an up-move and it is occurring in a lower horizontal path. Hence, either of $s^{in^{'}}_{1}$ or $s^{in^{'}}_{2}$ must be occupied. If $s^{in^{'}}_{2}$ is occupied, the above argument can be repeated until one reaches the boundary to find $N_v - y(v_j) + \nicefrac{1}{2}$ horizontal bonds that go through $x(s^{in}_2)$. If $s^{in'}_1$ is occupied, its path must be horizontal and running parallel to $q$, in particular making an up-step around $s^{in'}_1$. Again, we can continue the argument until we arrive at the boundary. The same argument can be used if initially $s^{in}_1$ is occupied instead of $s^{in}_2$. In either case, we find $N_v - y(v_j) + \nicefrac{1}{2}$ horizontal bonds that go through $x(s^{in}_2)$.\\
Now, by assumption, also either $s^{out}_1$ or $s^{out}_2$ is occupied. We can reverse top and bottom in the argument above to find another $y(v_j) + \nicefrac{1}{2}$ horizontal bonds which go through $x(s^{in}_2)$. Note, that the path that contains $s^{out}_2$ is necessarily upper, since otherwise, $m$ would be inside it and it could not exist yet by construction. \\
We have hence found $N_v + 1$ bonds in $p$ that go through a single vertical cut, contradicting the assumption that $p$ is allowed.
\item \emph{(Adding bubbles)} Now for each bond, we have created a connecting path. It is possible, however, that not all points in $\mathcal{L}_{N_h, N_v}$ are occupied. In this case we add small bubbles to the pattern.

\end{enumerate}
It remains to show that the loop pattern thus created is is compatible with $p$, i.e. that the boundary points of all paths correspond to tuples in the connectivity pattern, or - differently phrased - for a tuple $(a,b) \in p$, whether the corresponding path in $L$ starting with $a$ can end at a point $b' \neq b$. By construction, once $(v_j,b)$ becomes diagonal, it will surely have the correct ending point. Again, let us consider without loss of generality a horizontal path. Then, after $\Delta x$ steps, the horizontal distance to the target is zero. Hence, either we have arrived at the correct ending point, or the partial path has become vertical during the construction. To become vertical, however, the path must have gone through a point where its remainder was diagonal, hence ensuring that the correct ending point was reached. \\
The resulting loop pattern is the \emph{canonical loop pattern of p}.
\end{definition}

\section{The number $\mathcal{N}(N_h, N_v)$}
\label{sec:numberN}
Now that we have seen that there is at least one loop pattern for each allowed connectivity pattern, we can count the forbidden connectivity patterns.

\subsection{We can count horizontally and vertically forbidden connectivity patterns individually}
\begin{claim}A connectivity pattern cannot be both horizontally and vertically forbidden.
\end{claim}

\begin{proof} Let $p$ be a connectivity pattern and assume it is both horizontally and vertically forbidden. Then denote the vertical lines at which there is an oversaturated cut by $x$ and $y$, respectively. These lines cut the patch into four areas, $A, B, C$ and $D$ as depicted in figure \ref{fig:CutABCD}.
\begin{figure}[!ht]
  \caption{The boundary cut into regions by vertical and horizontal cuts at $x$ and $y$}
  \centering
  \label{fig:CutABCD}
    \input{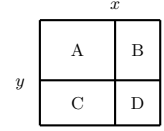}
\end{figure}
Now if each of the bonds cuts only either the horizontal \emph{or} vertical line, then there would need to be at least $N_h + N_v + 2$ bonds in total, hence at least two bonds cut both lines, without loss of generality going from boundary $A$ to boundary $D$ in the figure. There could be more than two bonds crossing from $A$ to $D$ - let us denote the total number by $\kappa$, the lowest one by $a$ and the highest one by $b$. These bonds partition the areas $A$ and $D$ into $A_L$, $A_R$ and $D_L$, $D_R$, respectively. For their size, clearly
\begin{equation}
\label{eq:AL}
\begin{aligned}
|A_L| + |A_R| +\kappa &\leq |A| \\
|D_L| + |D_R| +\kappa &\leq |D|
\end{aligned}
\end{equation}
holds.
There remain $N_h-\kappa +1$ bonds to be found for the horizontal violation and all of these must have boundary points in $A_R$ and $D_R$. Similarly, there remain $N_v-\kappa +1$ bonds to be found for the vertical violation and all of these must have boundary points in $A_L$ and $D_L$. Hence we have the inequalities
\begin{equation}
\begin{aligned}
|A_L| + |D_L| &\geq N_v -\kappa +1 \\
|A_R| + |D_R| &\geq N_h -\kappa +1.
\end{aligned}
\end{equation}
Adding the two inequalities and inserting inequalities (\ref{eq:AL}), we obtain 
\begin{equation}
\begin{aligned}
|A| + |D| \geq N_v + N_h + 2
\end{aligned}
\end{equation}
and since $|A| + |D| = N_v + N_h$, we arrive at a contradiction.
\end{proof}

\subsection{A bijection between connectivity patterns and Dyck Paths}
\begin{definition}\textbf{(Dyck paths)}\\
A \emph{Dyck path} or \emph{mountain diagram} of size $n$ is a lattice path in $\mathbb{Z}^2$ from $(0,0)$ to $(2n,0)$ consisting of $n$ up steps of the form $(1,1)$ and $n$ down steps of the form $(1,-1)$ which never goes below the x-axis $y=0$. The \emph{maximal height} of a Dyck path is the  maximum $y$-coordinate of the path. Denote all Dyck paths of size $n$ by $\mathcal{D}_n$.
\end{definition}

\begin{definition}\textbf{(Bijection between connectivity patterns and Dyck paths)} \label{def:bijection} \\
We define two maps
\begin{equation}
\begin{aligned}
\phi_h: (N_h,N_v)\text{-connectivity patterns} & \mapsto \mathcal{D}_{N_h + N_v} \\
\phi_v: (N_h,N_v)\text{-connectivity patterns} &\mapsto \mathcal{D}_{N_h + N_v} \\
\end{aligned}
\end{equation}
The image of a given connectivity patterns $p$ under the map $\phi_h$ is given as follows. We start with the empty Dyck path and sequentially look at the boundary points in the order given in figure \ref{fig:phiHorizontal}. Then, we add an up-step to the Dyck path if the partner of the boundary point we are currently reading has not been read yet. Otherwise we add a down-step. For $\phi_v$, we follow the same procedure with the labelling given by figure \ref{fig:phiVertical} instead.
\begin{figure}[!ht]
  \caption{The mapping $\phi_h$ for an allowed connectivity pattern}
  \centering
  \label{fig:phiHorizontal}
    \input{tikzpictures/mappingPHLeft} $\xrightarrow{\phi_h}$ \input{tikzpictures/mappingPHRight}
\end{figure}

\begin{figure}[!ht]
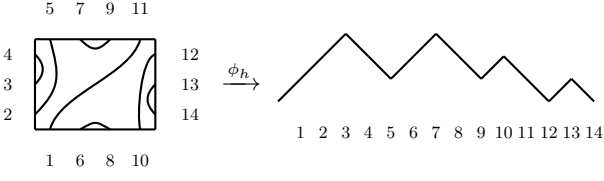

  \caption{The mapping $\phi_v$ for a forbidden connectivity pattern}
  \centering
  \label{fig:phiVertical}
   \input{tikzpictures/mappingPVLeft} $\xrightarrow{\phi_v}$ \input{tikzpictures/mappingPVRight}
\end{figure}
A couple of remarks are in order
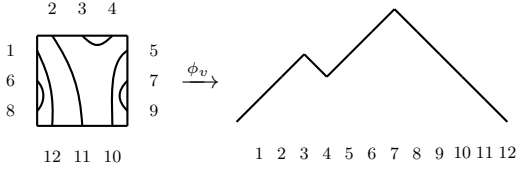
\begin{itemize}
\item The resulting path is a Dyck path: For it to pierce through the x-axis, one would need to read more second halves than first halves up to a given point which is clearly impossible. Also, there is an equal number of second halves and first halves in total, so the final step ends up on the x-axis again.
\item The maps $\phi_h$ and $\phi_v$ are bijective. The map $\phi_h^{-1}$ reads the Dyck path sequentially from start to end, while scanning through the boundary points in the order given in figure \ref{fig:phiHorizontal}. Whenever a down-step is encountered, a bond is added to the connectivity pattern by matching the currently active boundary point with the last open one. Again, $\phi_v^{-1}$ works analogously with the labelling given in figure \ref{fig:phiVertical}.
\item For $i \in \{1,2, \dots N_h-1 \}$, $Flow(p,i,vert)$ is given by the height of $\phi_h(p)$ after $N_v+2i$ steps. Similarly, for $j \in \{1,2, \dots N_v-1 \}$, $Flow(p,j,hor)$ is given by the height of $\phi_v(p)$ after $N_h+2i$ steps. In particular, $p$ is horizontally (vertically) forbidden if the maximal height of $\phi_h(p)$ ($\phi_v(p)$) is greater than $N_h$ ($N_v$).
\end{itemize}
\end{definition}

\subsection{Three expressions for the number of Dyck paths with restricted height}
\begin{claim}\textbf{(Number of allowed matchings)}
Let $C_n = \frac{1}{n+1}\binom{2n}{n}$ be the regular Catalan number. For a given $N_h, N_v \in \mathbb{N}$, the number of horizontally (vertically) forbidden connectivity pattern is given by $C_{N_h + N_v} - f(N_h, N_v)$ $(C_{N_h + N_v} - f(N_v, N_h))$, where we can give three expressions for the numbers $f(N_h, N_v)$:
\begin{widetext}
\begin{equation}
\begin{aligned}
\label{eq:modifiedCatalan}
f(N_h, N_v) &= \frac{4^N}{1+\frac{N_h}{2}} \sum_{j=1}^{N_h+2} \sin\left(\frac{\pi j}{N_h+2}\right) \left( \cos\left(\frac{\pi j}{N_h+2}\right)\right)^{2N} \\
&= \sum_{k \geq 1} \binom{2N}{N- k(N_h+2) -1} - 2\binom{2N}{N- k(N_h+2)} + \binom{2N}{N- k(N_h+2) +1} \\
&= \left(\frac{\text{d}}{\text{d}z}\right)^{2N}\Bigg|_{z=0}  \quad \cfrac{1}{1 - \cfrac{z^2}{\ddots \, - \cfrac{1}{z^2}}}
\end{aligned}
\end{equation}
\end{widetext}
where the continued fraction has $N_h$ instances of $z^2$ and $N=N_h+N_v$.
As a direct corollary, since forbidden matchings are either horizontally or vertically forbidden and the number of all matchings is $C_{N_h + N_v}$, we obtain the total number of allowed matchings:
\begin{equation}
\label{eq:exactForm1}
\begin{aligned}
\mathcal{N}(N_h, N_v) = f(N_h, N_v) + f(N_v, N_h) - C_{N_h + N_v}.
\end{aligned}
\end{equation}
For later convenience, we will introduce $f(N_h, \alpha)$, with the aspect ratio $\alpha N_v = N = N_h + N_v$. Since the number of vertically forbidden matchings is equal to the number of horizontally forbidden matchings on a $90^\circ$ rotated patch and a $90^\circ$ rotation corresponds to $\alpha \rightarrow \frac{\alpha}{\alpha-1}$, we can rewrite equation (\ref{eq:exactForm1}) as
\begin{equation}
\label{eq:exactForm2}
\begin{aligned}
\mathcal{N}(N_h, \alpha) = f(N_h, \alpha) + f\left(N_h, \frac{\alpha}{\alpha-1}\right) - C_{\alpha N_h}.
\end{aligned}
\end{equation}
\end{claim}

\begin{proof}
Let $N_h, N_v \in \mathbb{N}$. We are going to count the horizontally allowed connectivity pattern and show that they are equal to $C_{N_h + N_v} - f(N_h, N_v)$. From the definition of Dyck paths, we need to count all mountain diagrams of half-length $N_h + N_v$ whose maximum height exceeds $N_h$. To this end, we set up a sequence of counting vectors $v_i \in \mathbb{N}^{N_h+1}$. After $n$ steps, we would like the number of paths with height $h$ that never exceed $N_h$ in height to be given by $(v_n)_h$. Hence, set $v_0 = (1, 0, \dots, 0)$, indicating a single path with height 0, the empty path. Now we are going to sequentially apply linear operations
\begin{equation}
\begin{aligned}
v_{i+1} = M v_i,
\end{aligned}
\end{equation}
for $i \in \{1,\dots, 2N \}$, where $M$ is an $N_{h+1} \times N_{h+1}$-matrix, defined as
\begin{equation}
\begin{aligned}
M = 
\begin{bmatrix}
0  & 1     &       &               \\
1  & \ddots & \ddots &          \\
 &  \ddots& \ddots & 1  \\
  &        &  1& 0  \\
\end{bmatrix}
\end{aligned}
\end{equation}
For each existing path of length $n$ and height $h$, the matrix $M$ has the effect of creating two new paths of height $h+1$ and $h-1$, while automatically cutting off paths with height greater than $N_h$ and smaller than $0$. \\
Finally, the vector $v_{2N_h+2N_v}$ contains the number of paths of all heights after $2N_h+2N_v$ steps, out of which we are only interested in proper Dyck paths, the number of which is stored in $(v_{2N})_1=(M^{2N}v_0)_1$.
The matrix $M$ is a tridiagonal Toeplitz matrix with eigenvalues $D_j = 2 \cos\left(\frac{\pi j}{N_h+2} \right)$ and eigenvectors $S_{ij}= \frac{1}{\sqrt{1+\frac{N_h}{2}}} \sin \left(\frac{\pi i j}{N_h+2} \right)$. Since $S$ is orthogonal,
\begin{equation}
\begin{aligned}
(M^{2N}v_0)_1 &= \sum_{j=1} ^{N_h+2} S_{1j} D_j^{2N} (S^{-1})_{j1} \\
&= \sum_{j=1} ^{N_h+2} (S_{1j})^2 D_j^{2N} \\
&= f(N_h, N_v)
\end{aligned}
\end{equation}
The second form is an application of the combinatorics of \textit{watermelons} \cite{watermelons}.
For the last expression, we allude to a tool from analytical combinatorics, the \textit{symbolic method} \cite{symbolicmethod}. Assume that we want to calculate the number $b_n$ of binary words of length $n$. Then we can write down a combinatorial equation
\begin{equation}
\begin{aligned}
B = \underbrace{\epsilon}_{\text{empty word}} \cup \underbrace{B \times 0}_{\text{append a zero}} \cup \underbrace{B \times 1}_{\text{append a one}}
\end{aligned}
\end{equation}
meaning ``A binary word is either the empty word or a binary word ending on zero or a binary word ending on one''. The machinery of the symbolic method teaches us to translate this equation into a generating function
\begin{equation}
\begin{aligned}
B(z) &= 1 +zB(z) + zB(z) \quad \Rightarrow \\
B(z) &= \frac{1}{1-2z} \\
&= \sum_{n \geq 0} 2^n z^n
\end{aligned}
\end{equation}
such that we can extract the numbers $b_n = (\text{d}/\text{d}z)^{2N} B(z)|_{z=0} $ from the coefficients of the Taylor series.
To calculate the number $D_{2n}$ of Dyck paths of length $2n$, we use \textit{the first passage decomposition}:
\begin{equation}
\begin{aligned}
D = \underbrace{\epsilon}_{\text{empty path}} \cup \underbrace{\uparrow \times D \times \downarrow D}_{\substack{\text{an up-step followed by a Dyck path,} \\ \text{a down-step and another Dyck path}}}
\end{aligned}
\end{equation}
meaning ``A Dyck path is either empty or an up-step follwed by a Dyck path, a down-step and another (possibly empty) Dyck path". Similarly, this translates to 
\begin{equation}
\begin{aligned}
D(z) &= 1 +zD(z)zD(z) \quad \Rightarrow \\
D(z) &= \frac{1-\sqrt{1-4z}}{2z} \\
&= \sum_{n \geq 0} C_n z^{2n}
\end{aligned}
\end{equation}
where $C_n$ is the $n$-th Catalan number. Finally, to obtain a generating function $D_h(z)$ for the number of Dyck paths with maximal height $h$, we start with $D_0(z) = 1$, since there is exactly one path with length zero: the empty path. Again, we decompose the path to the left and right of its first passage of zero:
\begin{equation}
\begin{aligned}
D_h(z) &= 1 +zD_{h-1}(z)zD_h(z) \quad \Rightarrow \\
D_h(z) &= \cfrac{1}{1 - \cfrac{z^2}{\ddots \, - \cfrac{1}{z^2}}} \\
&= \sum_{N_v \geq 0} f(h, N_v) z^{2(h + N_v)}
\end{aligned}
\end{equation}
\end{proof}

\subsection{Proof of (\ref{eq:asymptoticsNhNv})}

We will now show that
\begin{equation}
\mathcal{N}(N, \alpha) = \frac{4^N}{N^{3/2}} \left[ k(\alpha) + \mathcal{O}\left(\frac{1}{N}\right) \right]
\end{equation}
cf. equation (\ref{eq:exactForm2}) with the function $k(\alpha)$ given by
\begin{align}
k(\alpha) = \frac{\sqrt{\pi}}{2} + \frac{\sqrt{\pi}}{2} (\alpha-1)^{3/2} - \pi^{-1/2}
\end{align}

In order to calculate $\mathcal{N}(N_h, \alpha)$, it is sufficient to compute the asymptotic behaviour of $f(N, \alpha)$ since the vertically forbidden loop patterns can be transformed into horizontally forbidden ones under the $90^\circ$ rotation $\alpha \rightarrow \frac{\alpha}{\alpha-1}$. The asymptotic behaviour of the Catalan numbers is known to be $4^N/N^{3/2}\sqrt{\pi}$.
It remains to calculate the expression
\begin{align}
g(N_h,\alpha):=\sqrt{N_h} \sum_{j=1}^{N_h+2} \sin\left( \frac{\pi j}{N_h+2} \right) \cos \left( \frac{\pi j}{N_h+2} \right)^{2N}
\end{align}
First observe that the summand is symmetric around $j=\frac{N_h+2}{2}$ (if $N_h$ is odd, we can omit the $\left(\ceil{\frac{N_h+2}{2}}\right)$th from the sum as this term is exponentially small in $N_h$). Therefore
\begin{align}
g(N_h,\alpha):=2 \sqrt{N_h} \sum_{j=1}^{N_h+2/2} \sin\left( \frac{\pi j}{N_h+2} \right) \cos \left( \frac{\pi j}{N_h+2} \right)^{2\alpha N_h}
\end{align}

We will proceed with the computation of the sum in four steps. First, we will truncate the sum, using the exponential suppression of terms with $j$ on the order of $N_h$. Second, we will replace the cosine by a Gaussian. Third, we Taylor expand the sine and finally, we replace the sum by an integral that we can compute analytically. All of these approximations induce an error $\mathcal{O}(1/N)$.

Let us now establish the relevant claims:
\begin{definition}
\begin{equation}
e_1(N):=\sqrt{N} \sum_{j=\floor{N/\pi}}^{N/2} \sin^2\left (\frac{\pi j}{N} \right) \cos^{2\alpha N}\left( \frac{\pi j}{N} \right)
\end{equation}
\end{definition}

\begin{claim}\textbf{(Truncation of the sum)}
\begin{equation}
e_1(N) = \mathcal{O} \left( N^{3/2} 2^{-\alpha N} \right)
\end{equation}
\end{claim}

\begin{proof}
For $j$ in the interval $[\floor{N/\pi},N/2]$, for $N$ large enough we have that $j \pi /N \geq \pi/4$ and therefore
\begin{align}
e_1(N)&\leq\sqrt{N} \sum_{j=\floor{N/\pi}}^{N/2} \sin^2\left (\frac{\pi j}{N} \right) \cos^{2\alpha N}\left( \pi/4 \right) \nonumber \\
&=\sqrt{N} \sum_{j=\floor{N/\pi}}^{N/2} \sin^2\left (\frac{\pi j}{N} \right) 2^{-\alpha N} \nonumber \\
&\leq N^{3/2} 2^{-\alpha N}
\end{align}
\end{proof}

\begin{definition}
\begin{equation}
e_2(N):=\sqrt{N} \sum_{j=1}^{\floor{N/\pi}} \sin^2\left (\frac{\pi j}{N} \right) \left[ e^{-\alpha N \left( \frac{\pi j}{N} \right)^2} - \cos^{2\alpha N}\left( \frac{\pi j}{N} \right) \right]
\end{equation}
\end{definition}

\begin{claim} \textbf{(A cosine raised to a high power becomes a Gaussian)}
\begin{equation}
e_2(N) = \mathcal{O}\left( \frac{1}{N} \right)
\end{equation}
\end{claim}

\begin{proof}
For simplicity, define $x_j = \frac{\pi j}{N}$ and $M = 2 \alpha N$. Using $\cos(x) \geq 1-\frac{1}{2} x^2 > 0$ in the interval $x \in [0, 1]$ and the fact that each term in the sum is positive, we have:
\begin{align}
e_2(N) &\leq \sqrt{N} \sum_{j=1}^{\floor{N/\pi}} \sin^2(x_j) \left[ e^{-\frac{x_j^2}{2} M} - \left( 1-\frac{1}{2}x_j^2 \right)^M \right] \nonumber \\
&= \sqrt{N} \sum_{j=1}^{\floor{N/\pi}} \sin^2(x_j) \left[ e^{-\frac{x_j^2}{2} M} - \left( 1-\frac{\frac{1}{2}x_j^2 M}{M} \right)^M \right]
\end{align}
Since $\frac{1}{2}x_j^2 < 1$, we can use the inequality $\left( 1-\frac{\frac{1}{2}x_j^2 M}{M} \right)^M \geq e^{-\frac{x_j^2}{2} M \left(1 - \frac{1}{1-\frac{x_j^2}{2}} \right)} $, combined with $0< \sin(x_j) < x_j$ and $e^{-y} \geq 1-y$ to arrive at
\begin{align}
\label{eq_chain}
e_2(N) &\leq \sqrt{N} \sum_{j=1}^{\floor{N/\pi}} x_j^2 \left[ e^{-\frac{x_j^2 M}{2}}  - e^{-\frac{x_j^2}{2} M \frac{1}{1-\frac{x_j^2}{2}} } \right] \nonumber \\
&= \sqrt{N} \sum_{j=1}^{\floor{N/\pi}} x_j^2 e^{-\frac{x_j^2 M}{2}} \left[ 1 - e^{\frac{x_j^2}{2} M \left(1 - \frac{1}{1-\frac{x_j^2}{2}} \right)} \right] \nonumber \\
&\leq \sqrt{N} \sum_{j=1}^{\floor{N/\pi}}  x_j^2 e^{-\frac{x_j^2 M}{2}} \left( - \frac{x_j^2}{2}M \left(1 - \frac{1}{1 - \frac{x_j^2}{2}} \right)  \right) \nonumber \\
&\leq \frac{\sqrt{N} M}{2} \sum_{j=1}^{\floor{N/\pi}} x_j^6 e^{-\frac{x_j^2 M}{2}} \underbrace{\frac{1}{1-\frac{x_j^2}{2}}}_{\leq 2} \nonumber \\
&\leq N^{5/2} \frac{2 \alpha}{\pi} \frac{1}{\floor{N/\pi}} \sum_{j=1}^{\floor{N/\pi}} x_j^6 e^{-\frac{x_j^2 M}{2}}
\end{align}
The sum is the right Riemann sum of the function $f(x) = \pi^2 x^6 e^{-x^2 \alpha N \pi^2}$, with an error given by
\begin{align}
\label{eq_riemann1}
\left| \frac{1}{\floor{N/\pi}} \sum_{j=1}^{\floor{N/\pi}} x_j^6 e^{-\frac{x_j^2 M}{2}} - \int_0^1 \pi^2 x^6 e^{-x^2 \alpha N \pi^2} \text{d}x \right| \leq \frac{d_\text{max}}{2 \floor{N/\pi}},
\end{align}
where $d_\text{max}$ is the maximum of the derivative $f'(x)$ in the interval $[0, 1]$. A direct calculation reveals that 
\begin{align}
\label{eq_dmax1}
d_\text{max} = c (N\alpha)^{-5/2}
\end{align}
for some constant $c$. Plugging (\ref{eq_dmax1}) and (\ref{eq_riemann1}) into (\ref{eq_chain}) yields
\begin{align}
e_2(N) &= N^{5/2} 2\alpha \pi \int_0^1 x^6 e^{-x^2 \alpha N} \text{d}x + \mathcal{O}\left(\frac{1}{N}\right) \nonumber \\
&\leq N^{5/2} 2\alpha \pi \int_0^\infty x^6 e^{-x^2 \alpha N} \text{d}x + \mathcal{O}\left(\frac{1}{N}\right) \nonumber \\
&= N^{5/2} 2\alpha \pi \frac{15\sqrt{\pi}}{16} (N \alpha \pi^2)^{-7/2} + \mathcal{O}\left(\frac{1}{N}\right) \nonumber \\
&= \mathcal{O}\left(\frac{1}{N}\right) 
\end{align}
\end{proof}

\begin{definition}
\begin{equation}
\label{eq_definition_replace_sin}
e_3(N):=\sqrt{N} \sum_{j=1}^{\floor{N/\pi}} e^{-\alpha N \left( \frac{\pi j}{N} \right)^2}  \left[\sin^2\left (\frac{\pi j}{N} \right)  - \left(\frac{\pi j}{N} \right)^2 \right]
\end{equation}
\end{definition}

\begin{claim} \textbf{(Replacing the sine)}
\begin{equation}
e_3(N): = \mathcal{O}\left( \frac{1}{N} \right)
\end{equation}
\end{claim}

\begin{proof}
Taylor expanding the sine yields
\begin{align}
\label{eq_Taylor_sin}
\sin^2(x) = x^2 - \frac{1}{3} \cos(2 \xi) x^4
\end{align}
for some $\xi \in [0,x]$. Plugging (\ref{eq_Taylor_sin}) into (\ref{eq_definition_replace_sin}) and using a Riemann sum bound akin to (\ref{eq_riemann1}) leads to
\begin{align}
e_3(N) &\leq \sqrt{N} \sum_{j=1}^{\floor{N/\pi}} e^{-\alpha N \left( \frac{\pi j}{N} \right)^2}  \left(\frac{\pi j}{N} \right)^4 \nonumber \\
&\leq N^{3/2} \pi^2 \int_0^1 x^4 e^{-\alpha N x^2 \pi^2} \text{d}x \nonumber \\
&\leq N^{3/2} \pi^2 \int_0^\infty x^4 e^{-\alpha N x^2 \pi^2} \text{d}x \nonumber \\
&\leq N^{3/2} \pi^2 \frac{3}{8} \sqrt{\pi} (\alpha \pi^2 N)^{-5/2} \nonumber \\
&= \mathcal{O}\left( \frac{1}{N} \right)
\end{align}
\end{proof}

\begin{definition}
\begin{equation}
\label{eq_definition_remainder}
r(N):= \sqrt{N} \sum_{j=1}^{\floor{N/\pi}} x_j^2 e^{-\alpha N x_j^2}
\end{equation}
\end{definition}

\begin{claim} \textbf{(Computation of the integral)}
\begin{equation}
r(N): = \frac{\sqrt{\pi}}{4\alpha^{3/2}} + \mathcal{O}\left( \frac{1}{N} \right)
\end{equation}
\end{claim}

\begin{proof}
The usual bound for the Riemann sum (\ref{eq_riemann1}) implies
\begin{align}
r(N) &= \sqrt{N} \floor{N/\pi} \frac{1}{\floor{N/\pi}} \sum_{j=1}^{\floor{N/\pi}} x_j^2 e^{-\alpha N x_j^2} \nonumber \\
&= \sqrt{N} \floor{N/\pi} \pi^2 \int_0^1 x^2 e^{-\alpha N x^2 \pi^2} \text{d}x + \mathcal{O}\left( \frac{1}{N} \right)
\end{align}
We can extend the integral to infinity by noting that $x^2 \leq x e^{x^2}$:
\begin{align}
\int_1^\infty x^2 e^{-\alpha N x^2} \text{d}x &\leq \int_1^\infty x e^{1-\alpha N x^2} \nonumber \\
&=  \frac{e^{-\alpha N + 1}}{2(\alpha N - 1)} \nonumber \\
&= \mathcal{O}(e^{-N}),
\end{align}
implying that
\begin{align}
r(N) &= \sqrt{N} \floor{N/\pi} \pi^2 \int_0^\infty x^2 e^{-\alpha N x^2 \pi^2} \text{d}x + \mathcal{O}\left( \frac{1}{N} \right) \nonumber\\
&= \sqrt{N} \floor{N/\pi} \pi^2 \frac{\sqrt{\pi}}{4(\alpha N \pi^2)^{3/2}} + \mathcal{O}\left( \frac{1}{N} \right) \nonumber\\
&= \frac{\sqrt{\pi}}{4\alpha^{3/2}} + \mathcal{O}\left( \frac{1}{N} \right)
\end{align}
\end{proof}

\begin{corollary}
\begin{equation}
g(N_h,\alpha) = \frac{\sqrt{\pi}}{2\alpha^{3/2}} + \mathcal{O}\left( \frac{1}{N_h} \right)
\end{equation}
\end{corollary}

\begin{proof}
This follows directly from the four previous claims.
\end{proof}

\begin{corollary}
\begin{equation}
\label{eq:asymptoticFinal}
\mathcal{N}(\alpha, N) =  \frac{4^N}{N^{3/2}} \left[ k(\alpha) + \mathcal{O}\left(\frac{1}{N}\right) \right]
\end{equation}
\end{corollary}

\begin{proof}
For large $N$, we have
\begin{align}
\mathcal{N}(\alpha, N) &= 4^N \Bigg[N_h^{-3/2} g(N_h,\alpha) \nonumber \\ &+ N_h^{-3/2} g\left(N_h, \frac{\alpha}{\alpha-1} \right) - N^{-3/2} \pi^{-1/2} \Bigg] \nonumber \\
&= 4^N \Bigg[N_h^{-3/2} \frac{\sqrt{\pi}}{2\alpha^{3/2}} \nonumber \\ &+ N_h^{-3/2} \frac{\sqrt{\pi}}{2\frac{\alpha}{\alpha-1}^{3/2}} - N^{-3/2} \pi^{-1/2}\Bigg] \nonumber \\
&= \frac{4^N}{N^{3/2}} \Bigg[\frac{\sqrt{\pi}}{2} + \frac{\sqrt{\pi}}{2} (\alpha-1)^{3/2} \nonumber \\ & - \pi^{-1/2} + \mathcal{O}\left( \frac{1}{N}\right) \Bigg]
\end{align}
\end{proof}

\section{The matrix $\Psi_R$ in (\ref{eq_PsiQR}) is invertible}
\label{sec:Binvertible}
The matrix $\Psi_R: V_\text{matchings} \mapsto V_\text{loops}^B$ maps connectivity patterns on the boundary of a hole \textit{inside} the torus onto loop patterns on the complement of the hole. A priori, this map does not have to be invertible. We show here that for a torus much larger than the hole, $\Psi_R$ is invertible. More precisely, if the size of the rectangular hole is $L_h \times L_v$ and the torus is $N_h \times N_v$, then we require
\begin{equation}
\label{eq:11x9Step0}
\begin{aligned}
\min\{N_h, N_v \} > \frac{3}{2} (L_h + L_v)
\end{aligned}
\end{equation}

The kernel of $\Psi_R$ is non-empty if and only if for every connectivity pattern, there exists a loop pattern on $R$ that is compatible with it. The following procedure produces such a loop pattern for an arbitrary inside connectivity pattern. 

\begin{equation}
\label{eq:11x9Step1}
\begin{aligned}
\input{tikzpictures/11x9Step1}
\end{aligned}
\end{equation}
Since we work on the torus, we can draw the rectangle in the center of our lattice.
\begin{enumerate}
\item Close any nearest neighbours in a minimal way (as shown in the figures). Clearly, these cannot interfere with each other. This can be done within one tile from the hole.
\begin{equation}
\label{eq:11x9Step3}
\begin{aligned}
\input{tikzpictures/11x9Step3}
\end{aligned}
\end{equation}

\item Remove the connected pairs from the connectivity pattern. There must necessarily be at least one newly formed nearest neighbour pair. Connect these minimally, avoiding the bonds that are already closed. This can be done within two tiles from the hole.
\begin{equation}
\label{eq:11x9Step4}
\begin{aligned}
\input{tikzpictures/11x9Step4}
\end{aligned}
\end{equation}

\item Again remove the connected pairs from the connectivity pattern, creating new nearest neighbours. As long as there is enough space on the torus, these pairs can be closed minimally. For each nested bond, one more tile of space is needed
\begin{equation}
\label{eq:11x9Step6}
\begin{aligned}
\input{tikzpictures/11x9Step6}
\end{aligned}
\end{equation}
\item Fill the rest of the loop pattern arbitrarily
\end{enumerate}

For a hole of size $L_h \times L_v$, there can be at most $\ceil{\frac{L_h + L_v}{2}}$ nested bonds. Therefore, if  $\min\{N_h, N_v \} > \frac{3}{2} (L_h + L_v)$, then there is enough space in every direction for the above procedure to generate a compatible loop pattern.

\section{Uniqueness of the ground state within a given connectivity class}
\label{appendix:uniqueness}

Having defined the canonical loop pattern of a connectivity pattern, we can now prove that the Hamiltonian defined by (\ref{eq:parentH}) is ergodic in the sense that for every two loop patterns $L$ and $L'$ in the same connectivity class, there exists a sequence of surgery moves such that
\begin{equation}
\begin{aligned}
L' =  \sigma_M(\sigma_{M-1}(\dots \sigma_1(L))\dots)
\end{aligned}
\end{equation}

The algorithm to arrive at the canonical loop pattern from an arbitrary starting pattern from surgery moves only contains three steps:
\begin{enumerate}
\item Tadpoles and larger bubbles are cut off.
\item All paths are consecutively made as short as possible. Any path with non-minimal length must necessarily contain both vertical and horizontal \textit{bay-type} plaquettes (Fig. \ref{fig:canonical}b). This pair must necessarily contain a loop in their inside. The loop can be moved through the bay by three consecutive surgery moves. If the bays had previously been adjacent, the path is now shorter, otherwise the bays are now closer together. Therefore, any path can be made as short as possible. Note that any surgery move only acts on one path plus a surrounding loop so previously shortened paths will always stay shortest during the application of further elementary moves in this step.
\item Every path now exclusively consists of \textit{up}- and \textit{down}-moves, the order of which may still differ from the canonical loop pattern, i.e. the path might not run as close as possible to the north-west boundary of the patch (Fig. \ref{fig:canonical}c). For the pattern to be compatible with the same boundary matching, the area between the current and the desired trajectory for any given path must be filled with small bubbles. We are finished after moving all the bubbles through the appropriate bays.
\end{enumerate}

\begin{figure}[t]
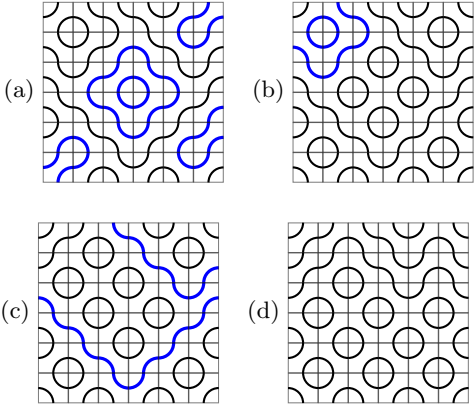

	(a) \input{tikzpictures/canonical0color} \, (b) \input{tikzpictures/canonical1color} \\ \vspace{5mm}
	(c) \input{tikzpictures/canonical2color} \, (d) \input{tikzpictures/canonical3color} 
	\caption{Bringing a given loop pattern (a) into the canonical pattern with the same boundary matching (d). In step 1, tadpoles and larger bubbles are cut off into small bubbles (b). Every path is shortened as much as possible in step 2 (c). The remaining ambiguity is the trajectory of longer paths. For the sake of uniqueness, they are moved as close as possible to the north-west boundary (d).}
	\label{fig:canonical}
\end{figure}

\section{Mapping to the Classical Delta Potts model}
\label{sec:PottsModel}

In this section, we will prove that certain correlation functions in the PEPS decay exponentially for $\lambda = 1$ (\ref{eq:definition}), by establishing a mapping to an observable in a classical model. We will start out with the case $u=\braket{0|1}_p = 0$ and comment on the general case later on.
For better readability, we denote the state at $\lambda = 1, u=0$ on a torus of size $N_h \times N_v$ with both $N_h$ and $N_v$ even by
\begin{equation}
\label{eq_}
\ket{\psi} := \ket{\psi_{N_h \times N_v} (u=0,\lambda=1)} = \sum_{\text{loop patterns } L} 2^{n_L} \ket{L}
\end{equation}

\begin{definition}
\begin{align}
\label{def_corrFunc}
\tilde{\sigma}_z(\vec{x}) := \begin{cases} \sigma_z(\vec{x}) &\mbox{if } \vec{x} \mbox{ is on the even sublattice} \\ 
-\sigma_z(\vec{x}) & \mbox{if } \vec{x} \mbox{ is on the odd sublattice} \end{cases}
\end{align}
A plaquette $\vec{x} = (x_1, x_2)$ is on the even (odd) sublattice if $x_1 + x_2$ is even (odd). The top left plaquette has coordinates (1,1).
\begin{align}
C[\vec{x}, \vec{y}] := \frac{\braket{\psi | \tilde{\sigma}_z(\vec{x}) \tilde{\sigma}_z(\vec{y})|\psi}}{\braket{\psi|\psi}} - \frac{\braket{\psi | \tilde{\sigma}_z(\vec{x})|\psi}\braket{\psi | \tilde{\sigma}_z(\vec{y}) | \psi}}{\braket{\psi|\psi}^2}
\end{align}
\end{definition}

\begin{claim}
\begin{align}
\label{eq_claim}
C[\vec{x}, \vec{y}] \xrightarrow{|\vec{x}-\vec{y}| \rightarrow \infty} e^{\frac{|\vec{x}-\vec{y}|}{\xi}}
\end{align}
for some $\xi > 0$.
\end{claim}

\begin{proof}
Consider a classical Q-state Potts model with spins residing on the vertices of the \textit{net lattice}. The net lattice is a square lattice rotated by $45^\circ$ where the distance between the vertices is increased by a factor $\sqrt{2}$ (the vertices are marked with green dots in figure \ref{fig:pottsPhases}). The classical spins take values $\sigma \in \{1, \dots, Q\}$. The Hamiltonian of the model is given by
\begin{equation}
\label{eq:pottsHamiltonian}
H = - \sum_{<ij>} \delta(\sigma_i, \sigma_j)
\end{equation}
where $<ij>$ indicates nearest neighbours on the net lattice. For a plaquette of the original square lattice located at $\vec{x}$, define by $\vec{x}_a$ and $\vec{x}_b$ the two spins adjacent to that plaquette (the order will not matter for our purposes).
Define the following ``link" observable in the Potts model that acts on two spins
\begin{align}
O_{\vec{x}}(\{ \sigma \}) :=  \begin{cases} 1 &\mbox{if } \sigma_{\vec{x}_a} = \sigma_{\vec{x}_b} \\
\frac{1+Q}{1-Q} &\mbox{if } \sigma_{\vec{x}_a} \neq \sigma_{\vec{x}_b} \end{cases}
\end{align}
Consider the expectation value of $O_{\vec{x}}$ in such a Potts model at inverse temperature $\beta$:
\begin{align}
\braket{O_{\vec{x}}} &= \frac{1}{Z} \sum_{ \{ \sigma \}} O_{\vec{x}}(\{ \sigma \}) \prod_{<ij>}e^{-\beta \delta(\sigma_i, \sigma_j)} \\
&= \frac{1}{Z} \sum_{ \{ \sigma \}} O_{\vec{x}}(\{ \sigma \}) \prod_{<ij>}[1 + v \delta(\sigma_i, \sigma_j)],
\end{align}
where $v = e^\beta-1$. We now expand the product in the spirit of the Fortuin-Kasteleyn expansion \cite{Wu,FortuinKasteleyn}, yielding $2^E$ terms, where $E$ is the number of edges of the net lattice.
\begin{align}
... &= \frac{1}{Z} \Big( \underbrace{\sum_{ \{ \sigma \}} O_{\vec{x}}(\{ \sigma \})}_{\input{tikzpictures/FKexpansion1}} + v \underbrace{\sum_{ \{ \sigma \}} O_{\vec{x}}(\{ \sigma \}) \delta(\sigma_1, \sigma_2)}_{\input{tikzpictures/FKexpansion2}} + \dots \Big) 
\end{align}

Here we have associated subgraphs $G'$ of the net lattice $G$ to each of the terms in the expansion, where $G'$ has an edge between $i$ and $j$ if the expansion term contains $\delta(\sigma_i, \sigma_j)$.
Let us investigate each sum individually. The first sum runs over $Q^4$ configurations. In $Q^3$ of those, $\sigma_{\vec{x}_a} = \sigma_{\vec{x}_b}$, implying $O_{\vec{x}}(\{ \sigma \})=1$. In the other $Q^3(Q-1)$ terms, $\sigma_{\vec{x}_a} \neq \sigma_{\vec{x}_b}$ and $O_{\vec{x}}(\{ \sigma \})=(1+Q)/(1-Q)$. Therefore, the sum evaluates to
\begin{align}
\sum_{ \{ \sigma \}} O_{\vec{x}}(\{ \sigma \}) = Q^3 + Q^3(Q-1)\frac{1+Q}{1-Q} = -Q^4
\end{align}
The second sum contains $Q^3$ configurations and because there is a $\delta$-function between spins 1 and 2, $\sigma_{\vec{x}_a} = \sigma_{\vec{x}_b}$ in all of them, 
\begin{align}
\sum_{ \{ \sigma \}} O_{\vec{x}}(\{ \sigma \}) \delta(\sigma_1, \sigma_2) = Q^3
\end{align}
Adding all these of contributions yields
\begin{align}
\label{eq_expValueG}
\braket{O_{\vec{x}}} = \frac{1}{Z} \sum_{G' \subseteq G} Q^{n(G')} v^{b(G')} \tilde{O}_{\vec{x}} (G')
\end{align}
where $n(G')$ is the number of connected components in $G'$, $b(G')$ is the number of bonds and 
\begin{align}
\tilde{O}_{\vec{x}} (G') = \begin{cases} 1 &\mbox{if } G' \mbox{ has a link at } \vec{x} \\ -1 &\mbox{ otherwise} \end{cases}
\end{align}
Each subgraph $G'$ of the net lattice can be associated to a unique loop pattern $L(G')$ on the square lattice (figure \ref{fig:pottsPhases}, such that for the number of closed loops we have
\begin{align}
\label{eq_loopOfG}
n_{L(G')} &= n(G') + c(G') \quad \text{and} \\
\tilde{O}_{\vec{x}} (G') &= \braket{L(G') | \tilde{\sigma}_z(\vec{x}) | L(G')},
\end{align}
where $c(G')$ is the number of circuits in $G'$.
Plugging Euler's relation 
\begin{align}
n(G') = c(G') - b(G') - V,
\end{align}
with $V$ the number of vertices in $G$ and (\ref{eq_loopOfG}) into (\ref{eq_expValueG}) yields
\begin{align}
\braket{O_{\vec{x}}} &= \frac{1}{Z} \sum_{G' \subseteq G} \sqrt{Q}^{n(G')}  \sqrt{Q}^{n(G')} v^{b(G')} \tilde{O}_{\vec{x}} (G') \nonumber \\
&= \frac{\sqrt{Q}^{-V}}{Z}  \sum_{G' \subseteq G} \sqrt{Q}^{n(G')+c(G')} \left(\frac{v}{\sqrt{Q}}\right)^{b(G')} \tilde{O}_{\vec{x}} (G') \nonumber \\
&=  \frac{\sum_{L} \sqrt{Q}^{n_L} \left(\frac{v}{\sqrt{Q}}\right)^{b(G')} \braket{L | \tilde{\sigma}_z(\vec{x}) | L}}{\sum_{L} \sqrt{Q}^{n_L} \left(\frac{v}{\sqrt{Q}}\right)^{b(G')}} \nonumber \\
&\xrightarrow{\beta \rightarrow \log(1+\sqrt{Q}) } \frac{\sum_{L} \sqrt{Q}^{n_L} \braket{L | \tilde{\sigma}_z(\vec{x}) | L}}{\sum_{L} \sqrt{Q}^{n_L}} \nonumber \\
&\xrightarrow{Q \rightarrow 16 } \frac{\braket{\psi | \tilde{\sigma}_z(\vec{x}) | \psi}}{\braket{\psi | \psi}}
\end{align}
This argument can be repeated for the correlator
\begin{align}
\braket{O_{\vec{x}} O_{\vec{y}}} = \frac{\braket{\psi | \tilde{\sigma}_z(\vec{x}) \tilde{\sigma}_z(\vec{y}) | \psi}}{\braket{\psi | \psi}},
\end{align}
thereby showing that
\begin{align}
\label{eq_correlatorEqPotts}
C[\vec{x}, \vec{y}] &= \braket{O_{\vec{x}} O_{\vec{y}}} - \braket{O_{\vec{x}}} \braket{O_{\vec{y}}},
\end{align}
i.e. the staggered $\sigma_z$ correlation function in the $\lambda =1$-PEPS equal to the link-link correlation of a classical $Q=16$ Potts model at $\beta = \log(1+\sqrt{Q})$. The model is known to undergo a phase transitions at that point for all values of $Q$. While this transition is critical for $Q \leq 4$ \cite{sqrt2barrier, Fendley1, Nienhuis, Qgroup1,Qgroup2,Qgroup3,Qgroup4,Qgroup5,FibonacciOrderFromNets,sqrt2barrier}, it is of first order for $Q>4$ \cite{PottsFirstOrder, PottsSecondOrder}, implying that the local correlator (\ref{eq_correlatorEqPotts}) decays exponentially, proving (\ref{eq_claim}).
 \end{proof}

A more general, alternative proof invokes the mapping between the norm of the PEPS to a Potts partition function. To this end, define the tensor network using tensors (\ref{eq:definition}) with independent variables on every site, i.e.:
\begin{align}
\ket{\psi(\lambda_{(1,1)}, \lambda_{(1,2)}, \dots, \lambda_{(N_h, N_v)})} =: \ket{\psi(\vec{\lambda})},
\end{align}
Taking derivatives with respect to different $\lambda$ will yield the expectation value of some local diagonal operator acting on e.g., one site, $D_\text{PEPS}(\vec{x})$
\begin{align}
\frac{\partial }{\partial \lambda_{\vec{x}}} \log \braket{\psi(\vec{\lambda}) | \psi(\vec{\lambda})} = \frac{\braket{\psi(\vec{\lambda}) | D_\text{PEPS}(\vec{x}) | \psi(\vec{\lambda})}}{\braket{\psi|\psi}}
\end{align}
Introducing the effective coupling strengths $\vec{\beta}$ via
\begin{align}
\lambda_{\vec{x}}  =  \begin{cases} \sqrt{\frac{e^{\beta_{\vec{x}}}-1}{\sqrt{Q}}} &\mbox{if } \vec{x} \mbox{ is on the even sublattice }  \\
 \sqrt{\frac{\sqrt{Q}}{e^{\beta_{\vec{x}}}-1}} &\mbox{if } \vec{x} \mbox{ is on the odd sublattice } \end{cases},
 \label{eq_pottsCouplings}
\end{align}
one can directly calculate that
\begin{align}
 \braket{\psi(\vec{\lambda}) | \psi(\vec{\lambda})} = C(\vec{\beta}) \underbrace{\sum_L \sqrt{Q}^{n_L} \prod_{\vec{x}} \left( \frac{e^{\beta_{\vec{x}}-1}}{\sqrt{Q}} \right)^{b_{\vec{x}}(L)}}_{Z_\text{inhom Potts}},
 \end{align}
where $b_{\vec{x}}(L)$ is 1 if $G'(L)$ has a bond at $\vec{x}$ and 0 otherwise. Here, $Z_\text{inhom Potts}$ is the partition function of a Potts model with different effective couplings between every pair of spins, given by (\ref{eq_pottsCouplings}). The constant is given by
\begin{align}
C(\vec{\beta}) = \prod_{\vec{x} \text{ odd}} \left( \frac{e^{\beta_{\vec{x}}}-1}{\sqrt{Q}} \right)
\end{align}
Therefore,
\begin{align}
\frac{\braket{\psi(\vec{\lambda}) | D_\text{PEPS}(\vec{x}) | \psi(\vec{\lambda})}}{\braket{\psi|\psi}} &= \frac{\partial \beta_{\vec{x}}}{\partial \lambda_{\vec{x}}} \frac{\partial }{\partial \beta_{\vec{x}}} \log \left( C(\vec{\beta}) Z_\text{inhom Potts} \right) \\
\end{align}
As usual, taking logarithmic derivatives of the partition function will yield some classical observable $D_\text{Potts}(\vec{x})$:
\begin{align}
... &= \braket{D_\text{Potts}(\vec{x})}
\end{align}
In particular, the point $\vec{\lambda} = \vec{1}$ corresponds to the original Potts model at its phase transition with all coupling strengths equal.
\begin{align}
\braket{D_\text{PEPS} (\vec{x})}_{\lambda=1} = \braket{D_\text{Potts}(\vec{x})}_{\beta = 1 + \sqrt{Q}},
\end{align}
Taking higher derivatives yields three-point and higher order correlators. In our case, $Q=16$ and all such operators decays exponentially even at the phase transition. This argument can be expanded by linearity to conclude that all diagonal correlators of the PEPS must decay exponentially.
 
Finally, for $u \neq 0$, the mapping has to be carried out with respect to two \textit{coupled} Potts models, whose phase diagram is also known \cite{coupledPotts1, coupledPotts2, coupledPotts3}. As the nature of the phase transition remains unchanged, we expect the correlation function to behave in the same manner as derived above.

\begin{figure}[t]
	(a) \input{tikzpictures/pottsOrdered} \, (b) \input{tikzpictures/pottsDisordered} \\ \vspace{5mm}
	(c) \input{tikzpictures/PT1Order} \, (d) \input{tikzpictures/PT2Order}
	\caption{Typical configurations in the Fortuin-Kasteleyn expansion
of the partition function of the Potts model. The green lines correspond
to the clusters in the expansion. Each cluster configuration is associated
with a unique loop pattern.	(a) A typical configuration of the
Potts model in the ordered phase, (b) a typical configuration in the
disordered phase, (c) the Potts model at the phase transition point for $Q
> 4$, (d) for values $Q \leq 4$. Only for the latter, loops of all length
scales occur, whereas the bounded loop length in all other cases
corresponds to a finite correlation length. }
	\label{fig:pottsPhases}
\end{figure}

\section{Open boundary conditions and unique ground state}
\label{app:obc-unique}
In section \ref{sec_OBCMain} we prove that the ground state of a modified parent Hamiltonian is unique. The key step in the proof is the fact that the minimal connectivity pattern can be reached from arbitrary starting loop patterns.

\begin{claim}
For every loop pattern $L$, there exists a sequence $\Sigma$ of bulk moves (\ref{eq:randomWalk}) and boundary moves (\ref{eq:randomWalk2}), such that
\begin{equation}
p(\Sigma(L)) = p_\text{min},
\end{equation}
where $p_\text{min}$ is given by (\ref{eq_pMinimal}).
\end{claim}
\begin{proof}
We are going to construct $\Sigma$ explicitly, starting from an arbitrary loop pattern $L$. We begin in the top left corner. Combining boundary moves on the first horizontal and vertical dominos and potentially a bulk move on the plaquette in the top left corner, we can transform the top left corner of $L$ into
\begin{equation}
\label{eq:fig2b}
\begin{aligned}
\input{tikzpictures/uniqueness06} 
\end{aligned}
\end{equation}
We proceed similarly for all other corners:
\begin{equation}
\label{eq:fig2c}
\begin{aligned}
 \input{tikzpictures/uniqueness07}
\end{aligned}
\end{equation}

Now we continue sequentially, column by column. If the top domino looks like $\input{tikzpictures/top00small}$ or $\input{tikzpictures/top11small}$, transform it into $\input{tikzpictures/top10small}$ using (\ref{eq:randomWalk2}). If it is in the $\input{tikzpictures/top01small}$-state, we will see now that the corresponding plaquette can be brought into the $\ket{B}$ state, after which the bubble is cut off and the boundary move is applicable again.

There are two scenarios:  In the first, the top $\input{tikzpictures/top01small}$-domino has a bubble or tadpole underneath it:
\begin{equation}
\begin{aligned}
\label{eq:uniqueStep1}
\input{tikzpictures/uniqueness08}
\end{aligned}
\end{equation}
In this case, the bubble can be moved up to the topmost plaquette using bulk moves. It can then be cut off to transform the top domino into $\input{tikzpictures/top10small}$. In the second scenario there is no bubble or tadpole in the column:
\begin{equation}
\begin{aligned}
\label{eq:uniqueStep2}
\input{tikzpictures/uniqueness09}
\end{aligned}
\end{equation}
Then, there must necessarily be $N_v$ paths passing through the column left to right. Only $N_v-4$ of them can originate from the west boundary, since $4+4k$ out of the $N_v+4k$ boundary points to the left of the $k$-th column are already connected with their nearest neighbours. Therefore, at least two pairs of paths must actually be a single path, which has a tadpole to the left of the column. This tadpole can be moved into the column upon which we recover situation 1, e.g.:
\begin{equation}
\begin{aligned}
\label{eq:uniqueStep3}
\input{tikzpictures/uniqueness10} \rightarrow \input{tikzpictures/uniqueness11}
\end{aligned}
\end{equation}
The bottom tile is transformed into $\input{tikzpictures/top01small}$ in the same manner to arrive at
\begin{equation}
\begin{aligned}
\label{eq:uniqueStep13}
\input{tikzpictures/uniqueness13}  
\end{aligned}
\end{equation}
After fixing the top and bottom dominos column by column, we apply the same procedure to the left and right boundary. Evidently, once a boundary domino is in the correct state (e.g., $\input{tikzpictures/top10small}$ for top dominos), it will never be touched again during this procedure, allowing us to sequentially bring the connectivity pattern into minimal form.
\end{proof}

\section{Dimension of the string-inserted subspace}
\label{app:strings}
Let $N_h$, $N_v$ be even and define
\begin{equation}
\begin{aligned}
\label{eq:app_stringInsertedPsi}
\ket{\psi \{U,V\}} := \, \input{tikzpictures/string_inserted2}
\end{aligned}
\end{equation}
where the boxes are $A$-tensors defined in eq. (\ref{eq:definition}) and periodic boundary conditions are enforced such that the tensor network lives on an $N_h \times N_v$-torus (we have dropped the subscript indicating the system size for better readability in the following). Note that one has to complex conjugate every other unitary in order for one to be able to pull the strings through a row or column of tensors respectively. If $U$ and $V$ commute, these strings can be moved freely through the system and it follows that $H \ket{\psi \{U,V\}} = 0$ for the parent Hamiltonian defined in (\ref{eq:parent-ham-def}). The purpose of this section is to compute the dimension of this \textit{string-inserted} subspace of the ground state manifold:
\begin{align}
\mathcal{S}' := \text{span}\{ \ket{\psi \{U,V \}} | [U,V] = 0, U,V \in SU(2) \}
\end{align}
We will show, that 
\begin{align}
\text{dim}\mathcal{S}' = \frac{(N_h + 1)(N_v+1) + 1}{2}
\end{align}
To this end, it is useful to make the following definition:

\begin{definition} For a tuple of $(j,k) \in \mathbb{Z}^2$, define
\begin{align}
g(j,k) &= \begin{cases} gcd(j,|k|) &\mbox{if } j,k \neq 0 \\
j &\mbox{if } k = 0, j>0 \\
|k| &\mbox{if } j = 0, |k|>0 \\
1 &\mbox{if k=j=0} 
\end{cases}
\end{align}
\end{definition}

We make the following observations:
\begin{itemize}
\item The winding number of a non-trivial loop in, say, the horizontal direction is equivalent to the difference of how many times that loop crosses the $U$-subset of the right boundary vs. how many times it crosses the $\bar{U}$-subset of the right boundary. These are the odd and even points on the boundary, respectively, as shown in (\ref{eq:app_stringInsertedPsi}). An equivalent statement holds for non-trivial winding in the vertical direction.
\item If in a given loop pattern $L$ there is a loop winding non-trivially around the torus $j$ times in the horizontal direction and $k$ times in the vertical direction, then \textit{all} non-trivial loops have winding number $(j,k)$ or $(-j,-k)$ (in fact, half of the loops will have winding number $(j,k)$ and the other half $(-j,-k)$). Therefore, we may denote the winding sector of such a loop pattern by $W(L) = (j,k) $. To remove ambiguity, we enforce $j \geq 0$.
\item A loop cannot wind around the torus $(j,k)$ times if $g(j,k) \neq 1$. 
\item A loop pattern in a given winding sector $(j,k)$, can have $n_{NTL} \in \{2, 4, \dots, \min \{ \floor{N_h/j}, \floor{N_v/k} \} \}$ non-trivial loops. Therefore we redefine the winding sector of a loop pattern that has $n_{NTL}/2$ loops wrapping around the torus $(j,k)$ times and $n_{NTL}/2$ loops wrapping around the torus $(-j,-k)$ times as $W(L) = (j \times n_{NTL}/2,k \times n_{NTL}/2)$.
\end{itemize}

We are now ready for some helpful definitions and shorthand notations:
\begin{definition}
\begin{align}
\widetilde{N_h} &:= N_h + 1 \\
\widetilde{N_v} &:= N_v + 1
\end{align}
\end{definition}

\begin{definition}
\begin{align}
\label{def_jk}
\ket{j,k} := \sum_{\substack{L \text{ s.t. } \\ W(L) = (j,k)}} 2^{n_L} \ket{L}
\end{align}
By the orthogonality of the physical basis states, different $\ket{j,k}$ are clearly orthogonal:
\begin{align}
\braket{j,k | j',k'} := \delta_{jj',kk'} || \ket{j,k} ||^2
\end{align}
and by the above observations $|| \ket{j,k} ||^2 \neq 0$ for all $(j,k) \in I$.
\end{definition}

\begin{definition}
\begin{align}
D_\phi = \begin{pmatrix} e^{i\phi} & 0 \\ 0 &  e^{-i\phi}  \end{pmatrix}
\end{align}
\begin{align}
W_{N_v}(\phi) :&= \bigotimes_{i=1}^{N_v/2} D_\phi \otimes \bar{D}_\phi \\
\tilde{W}_{N_v}^l :&= W_{N_v}\left(\frac{\pi l}{\widetilde{N_v}}\right)
\end{align}
\end{definition}

\begin{definition}
\begin{align}
\ket{\psi_{\phi, \theta}} &:= \ket{\psi \{D_\phi, D_\theta \}} \\
\ket{\tilde{\psi}_{l, m}} &:= \ket{\psi_{\phi = \frac{\pi l}{\widetilde{N_v}}, \theta = \frac{\pi m}{\widetilde{N_h}}}}
\end{align}
\end{definition}

\begin{definition}
\begin{align}
I = \Bigg \{ (x,y) | & x = 0, \dots, N_v/2, \nonumber \\ & y= \begin{cases} 0, \dots N_h/2 &\mbox{if } x = 0 \\ 
-N_h/2, \dots, N_h/2 & \mbox{if } x \neq 0 \end{cases} \Bigg \} 
\end{align}
Counting the number of elements in $I$ reveals that
\begin{align}
\label{def:I}
|I| = \frac{(N_h + 1)(N_v+1) + 1}{2}
\end{align}
\end{definition}

\begin{definition}
\label{def_phikxky}
\begin{align}
\ket{\phi_{k_x, k_y}} := & \frac{1}{\widetilde{N_v}\widetilde{N_h}} \times \nonumber \\
 \sum_{l=0}^{N_v} &\sum_{m=0}^{N_h} e^{2 \pi i \left(\frac{k_x l}{\widetilde{N_v}} + \frac{k_ym}{\widetilde{N_h}} \right)} \ket{\tilde{\psi}_{l, m}}
\end{align}
\end{definition}

\begin{definition}
\begin{align}
\label{def_Mjklm}
M_{(jk),(lm)} := \left[2 \cos \left( \frac{\pi j l}{g(j,k)\widetilde{N_v}} + \right.\right. \nonumber \\
\left. \left. \frac{\pi k m}{g(j,k)\widetilde{N_h}} \right) \right]^{2g(j,k)}
\end{align}
for any set of integers $j,k,l$ and $m$.
\end{definition}

\begin{claim}
\begin{align}
\label{claim1}
\dim \mathcal{S}' \leq |I|
\end{align}
\end{claim}

\begin{claim}
\begin{align}
\label{claim2}
\dim \mathcal{S}' \geq |I|
\end{align}
\end{claim}

Together, (\ref{claim1}), (\ref{claim2}) and (\ref{def:I}) entail that
\begin{align}
\dim \mathcal{S}' = \frac{(N_h + 1)(N_v+1) + 1}{2}
\end{align}

\begin{proof}{(of (\ref{claim1}))}
Starting from the definition of $\mathcal{S}'$, we can first restrict the unitaries $U$ and $V$ to be diagonal, i.e. 
\begin{align}
\mathcal{S}' = \text{span} \{ \ket{\psi_{\phi, \theta}} | \phi, \theta \in [0, 2\pi] \}
\end{align}
This is because any state that is generated by non-diagonal $U$ and $V$ is related to a state with $U$ and $V$ diagonal by conjugating the whole network with $S$, where $S$ is the unitary that simultaneously diagonalises $U$ and $V$. Because of the fundamental symmetry of the PEPS tensor, this conjugation leaves the state invariant.

As a first step, we are going to show that
\begin{align}
\label{reduce_to_lm}
\mathcal{S}' = \text{span}\{  \ket{\tilde{\psi}_{l, m}}\}_{\substack{ l = 0, \dots N_v \\ m=0, \dots N_h }}
\end{align}
Because $\ket{\psi_{\phi, \theta}}$ depends linearly on $W_{N_v}(\phi) \otimes W_{N_h}(\theta)$, it is sufficient to show that
\begin{align}
&\text{span} \{ W_{N_v}(\phi) \otimes W_{N_h}(\theta) \} = \nonumber \\
& \text{span} \left \{ \tilde{W}_{N_v}^l \otimes \tilde{W}_{N_h}^m  \right \}_{\substack{ l = 0, \dots N_v \\ m=0, \dots N_h }}
\end{align}
Clearly,
\begin{align}
\label{eq:inclusion}
&\text{span} \{ W_{N_v}(\phi) \otimes W_{N_h}(\theta) \} \supseteq \nonumber \\
& \text{span} \left \{ \tilde{W}_{N_v}^l \otimes \tilde{W}_{N_h}^m \right \}_{\substack{ l = 0, \dots N_v \\ m=0, \dots N_h }}
\end{align}
and we will prove the reverse inclusion by showing that
\begin{align}
\label{eq:inequalities}
\widetilde{N_v}\widetilde{N_h} & \geq \dim \text{span} \{ W_{N_v}(\phi) \otimes W_{N_h}(\theta) \} \nonumber \\
& \geq \dim \text{span} \left \{ \tilde{W}_{N_v}^l \otimes \tilde{W}_{N_h}^m \right \}_{\substack{ l = 0, \dots N_v \\ m=0, \dots N_h }}
  \nonumber \\ 
& \geq \widetilde{N_v}\widetilde{N_h}
\end{align}

The first inequality of (\ref{eq:inequalities}) follows by expanding the operator
\begin{align}
W_{N_v}(\phi) &= e^{iN_v \phi} \mathds{1}_{\binom{N_v}{0}} \oplus e^{i (N_v - 2) \phi} \mathds{1}_{\binom{N_v}{1}} \oplus \nonumber \\
 & \dots \oplus e^{-iN_v \phi} \mathds{1}_{\binom{N_v}{N_v}},
\end{align}
which, for general values of $\phi$ spans an $\widetilde{N_v}$-dimensional space.

The second inequality in (\ref{eq:inequalities}) is a trivial conclusion of (\ref{eq:inclusion}).

To see the validity of the third inequality, consider the matrix whose columns are made up of the distinct diagonal entries of $\tilde{W}_{N_v}^l$ for $l=\frac{N_v}{2}, \frac{N_v}{2}-1, \dots, 0, N_v, N_v-1, \dots, \frac{N_v}{2}+1$:
\begin{align}
F_{N_v} &= \begin{pmatrix}  \vert & \vert &  & \\ \text{diag} (\tilde{W}_{N_v}^{N_v/2}) & \text{diag}(\tilde{W}_{N_v}^{N_v/2-1}) & \dots  \\ \vert & \vert &  &  \end{pmatrix} \nonumber \\
&= \begin{pmatrix} 1 & 1 & 1 & 1 & \\ 1 &  \omega & \omega^2 & \omega^3 &  \\  1 &  \omega^2 & \omega^4 & \omega^6 & \dots \\  1 &  \omega^3 & \omega^6 & \omega^9 & \\ & & \vdots & & \ddots  \end{pmatrix} 
\end{align}
which is simply $\widetilde{N_v}$ times the $\widetilde{N_v} \times \widetilde{N_v}$ discrete Fourier matrix (we have set $\omega = \exp(2 \pi i /\widetilde{N_v}$)), and therefore has full rank equal to $\widetilde{N_v}$.

Applying these arguments to both tensor factors individually yields (\ref{eq:inequalities}).

Finally, we will prove that
\begin{align}
\label{reduce_to_I}
\text{span}\{  \ket{\tilde{\psi}_{l, m}}\}_{\substack{ l = 0, \dots N_v \\ m=0, \dots N_h }} = \text{span}\{  \ket{\tilde{\psi}_{l, m}}\}_{(l,m) \in I }
\end{align}
by showing that for each $(l,m) \notin I$, there exists an $(l',m') \in I$ such that $\ket{\tilde{\psi}_{l, m}} = \ket{\tilde{\psi}_{l', m'}}$.
The key observation is that 
\begin{align}
\ket{\psi\{U,V\}} &= \ket{\psi \{XUX^\dagger,XVX^\dagger\}} \\
\ket{\psi\{U,V\}} &= \ket{\psi \{-U,V\}} \\
\ket{\psi\{U,V\}} &= \ket{\psi \{U,-V\}}
\end{align}
which follows from the fact that conjugating the whole tensor network with $iX \in SU(2)$ leaves the state invariant and the numbers of $U$s and $V$s are both even. Inserting $U = \text{diag}(\exp(\pi i l/\widetilde{N_v}) , \exp(-\pi i l/\widetilde{N_v}))$ and $V = \text{diag}(\exp(\pi i m/\widetilde{N_h}) , \exp(-\pi i m/\widetilde{N_h}))$, we obtain
\begin{align}
\label{eq_relationslm1}
\ket{\tilde{\psi}_{l, m}} &= \ket{\tilde{\psi}_{-l, -m}}  \\
\ket{\tilde{\psi}_{l, m}} &= \ket{\tilde{\psi}_{l \pm \widetilde{N_v}, m}}  \\
\ket{\tilde{\psi}_{l, m}} &= \ket{\tilde{\psi}_{l, m \pm \widetilde{N_h}}}
\label{eq_relationslm3}
\end{align}
Using (\ref{eq_relationslm1}) - (\ref{eq_relationslm3}), for each $(l,m) \in [0,\dots N_v] \times [0, \dots N_h]$ we can now find an $(l',m') \in I$ such that $\ket{\tilde{\psi}_{l, m}} = \ket{\tilde{\psi}_{l',m'}}$
which imply (\ref{reduce_to_I}) and, together with (\ref{reduce_to_lm}) show that
\begin{align}
\dim \mathcal{S}' \leq |I|
\end{align}
\end{proof}

\begin{proof}{(of (\ref{claim2}))}
Because of (\ref{reduce_to_lm}) and the $\ket{\phi_{k_x, k_y}}$ being linear combinations of the $\ket{\tilde{\psi}_{l, m}}$ via definition \ref{def_phikxky}, it is clear that
\begin{align}
\label{T_supset}
\mathcal{S}' \supseteq \text{span} \{ \ket{\phi_{k_x, k_y}} \}_{(k_x, k_y) \in I}
\end{align}
Also, from the observations made in the beginning of this section and (\ref{def_jk}) and (\ref{def_Mjklm}), we see that
\begin{align}
\label{ket_lm_as_jk}
\ket{\tilde{\psi}_{l, m}} = \sum_{(jk) \in I} M_{(jk),(lm)} \ket{j,k}
\end{align}
The matrix elements of $M$ can be simplified using the binomial theorem. For better readability, we are going to suppress the argument of $g = g(j,k)$.
\begin{align}
\label{eq:matrix_massaged2}
M_{(jk),(lm)} &= \left[ 2 \cos\left(\frac{jl\pi}{\widetilde{N_v}g} + \frac{km\pi}{\widetilde{N_h}g}\right) \right]^{2g} \nonumber \\
&= \left[ e^{\frac{\pi i}{g}\left(\frac{jl}{\widetilde{N_v}} + \frac{km}{\widetilde{N_h}}\right)} + e^{-\frac{\pi i}{g}\left(\frac{jl}{\widetilde{N_v}} + \frac{km}{\widetilde{N_h}}\right)} \right]^{2g} \nonumber \\
&= \sum_{a=0}^{2g} e^{\frac{2 \pi i}{g}\left(\frac{jl}{\widetilde{N_v}} + \frac{km}{\widetilde{N_h}}\right)(a-g)} {2g \choose a}
\end{align}

Plugging (\ref{ket_lm_as_jk}) and (\ref{eq:matrix_massaged2}) into definition \ref{def_phikxky} yields
\begin{align}
\ket{\phi_{(k_x, k_y)}} &= \sum_{(jk)\in I} \, \sum_{a=0}^{2g} {2g \choose a} \ket{j,k} \nonumber \\
&\times \underbrace{\frac{1}{\widetilde{N_v}} \sum_{l=0}^{N_v} \left[ e^{\frac{2 \pi i}{\widetilde{N_v}} (\frac{j}{g}(a-g) - k_x )} \right]^l}_{\delta_{\frac{j}{g}(a-g)-k_x \in \widetilde{N_v} \mathbb{Z}}} \nonumber \\
&\times \underbrace{\frac{1}{\widetilde{N_h}} \sum_{m=0}^{N_h} \left[ e^{\frac{2 \pi i}{\widetilde{N_h}} (\frac{k}{g}(a-g) - k_y )} \right]^m}_{\delta_{\frac{k}{g}(a-g)-k_y \in \widetilde{N_h} \mathbb{Z}}} 
\label{readoff00}
\end{align}

In principle, the constraints only enforce e.g
\begin{align}
\label{eq_constraint}
\frac{k}{g}(a-g)-k_y = n\widetilde{N_h} 
\end{align}
for $n \in \mathbb{Z}$. However, we will now show that if $|n| \geq 1$, then it follows that $|a-g| > g$ which entails that either $a<0$ or $a>2g$, in both cases the summation on $a$ will be empty. Rearranging (\ref{eq_constraint}) and taking the absolute value yields
\begin{align}
|a-g| &= \frac{|n \widetilde{N_h} + k_y|}{|k|}g \nonumber \\
&\geq \frac{|n| \widetilde{N_h} - |k_y|}{|k|}g \nonumber \\
&> \widetilde{N_h} \frac{|n| -1/2}{|k|}g \nonumber \\
&> \frac{\widetilde{N_h}}{2} \frac{2}{\widetilde{N_h}}g \nonumber \\
&=g, \\
\end{align}
where we have used that $|a+b| > |a| - |b|$, $|n| \geq 1$, $|k_y| < N_h/2$ and $|k| < N_h/2$. This argument can be carried out for the constraints originating from both the summation over $l$ and $m$, leaving us with:
\begin{align}
\label{eq_finalconstraint}
\frac{j}{g}(a-g)&=k_x \\
\frac{k}{g}(a-g)&=k_y
\end{align}
These equations mean that
\begin{align}
\braket{j,k | \phi_{(k_x, k_y)}} = \begin{cases} {2g \choose k_x g/j + g} &\mbox{if } k_x/k_y = j/k \\ 
0& \mbox{otherwise}  \end{cases}
\end{align}
In particular, by orthogonality of the $\ket{j,k}$, sectors with different $k_x/k_y$ are mutually orthogonal. As a final step, we will investigate the sector that is spanned by the vectors
\begin{align}
\{ \ket{\phi_{(k_x, k_y)}} \}_{\substack{(k_x, k_y)\in I \\ k_x/k_y = p/q}}
\end{align}
for a fixed, completely reduced fraction $p/q$.  The vectors in this set have the form $\ket{\phi_{(p,q)}}, \ket{\phi_{(2p,2q)}}, \dots$. Since
\begin{align}
\braket{p,q | \phi_{(bp, bq)}} = {2 \choose b + 1},
\end{align}
only $\ket{\phi_{(1p, 1q)}}$ has non-zero overlap with $\ket{p,q}$. Therefore, $\ket{\phi_{(1p, 1q)}}$ must necessarily be linearly independent from all other vectors in that sector. We can therefore remove $\ket{\phi_{(1p, 1q)}}$ from $\{ \ket{\phi_{(k_x, k_y)}} \}_{(k_x, k_y)\in I \, k_x/k_y = p/q}$ and check the remaining basis vectors for linear independence. Indeed we can iterate this procedure to show that in the remaining set, there exists exactly one vector that has non-zero overlap with $\ket{bp,bq}$, which is $\ket{\phi_{(bp, bq)}}$. Therefore,
\begin{align}
\dim \text{span} \{ \ket{\phi_{(k_x, k_y)}} \}_{(k_x, k_y)\in I} = |I|
\end{align}
and by equation (\ref{T_supset}), it follows that
\begin{align}
\dim \mathcal{S}' \geq |I|
\end{align}

\end{proof}

\end{document}